%% file: thesis.tex
\documentclass[letterpaper, 12pt, oneside]{book}

    \usepackage[left=108pt, right=74pt, top=108pt, bottom=81pt, dvips, pdftex]{geometry}
    \usepackage{verbatim}
    \usepackage{amssymb}
    \usepackage{graphicx}
    \usepackage{graphics}
\usepackage[all]{xy}
    \usepackage{amsfonts}
    \usepackage{calc}
    \usepackage{amscd}
    \usepackage{color}
    \definecolor{darkblue}{rgb}{0, 0, .4}
\definecolor{grey}{rgb}{.7, .7, .7}
\usepackage[breaklinks]{hyperref}
\hypersetup{
        colorlinks=true,
        linkcolor=darkblue,
        anchorcolor=darkblue,
        citecolor=darkblue,
        pagecolor=darkblue,
        urlcolor=darkblue,
        pdftitle={},
        pdfauthor={}
}
\usepackage[vcentermath,noautoscale,enableskew]{youngtab}
\usepackage{yhmath}
    \usepackage{epsfig}
    \usepackage{fancyhdr} 
    \renewcommand{\bibname}{References}

    \usepackage{amsmath,amsthm, amsfonts,amssymb} 
    \usepackage{setspace} 

    \newtheorem{theorem}{Theorem}[section]

    \theoremstyle{definition}

        \newtheorem{example}[theorem]{Example}

    \theoremstyle{remark}
        \newtheorem{remark}[theorem]{Remark}
    
    \numberwithin{equation}{section}
    \numberwithin{figure}{section}

    \newcommand{\n}{\vspace{12pt}} 

    \newcommand{\newchapter}[3] 
        {                           
        \chapter[#2]{#3}
        \chaptermark{#1}
        \thispagestyle{myheadings}
        }

\newcommand{\C}{\mathbb{C}}
\newcommand{\F}{\mathbb{F}}
\newcommand{\R}{\mathbb{R}}

\newcommand{\eps}{\epsilon}

\newcommand{\gl}{\mathfrak{gl}}

\newcommand{\poly}{\mathrm{poly}}

\newcommand{\End}{\mathrm{End}}
\newcommand{\Res}{\mathrm{Res}}
\newcommand{\SYT}{\mathrm{SYT}}
\newcommand{\SSYT}{\mathrm{SSYT}}
\newcommand{\sh}{\mathrm{sh}}

\newcommand{\ie}{\emph{i.e.}}

\newcommand{\bra}[1]{\langle#1|}
\newcommand{\ket}[1]{|#1\rangle}

\newcommand{\tensor}{\otimes}

\newcommand{\redtwo}[0]{{\color{red}2}}
\newcommand{\redthree}[0]{{\color{red}3}}
\newcommand{\redfour}[0]{{\color{red}4}}

\newcommand{\bluetwo}[0]{{\color{blue}2}}
\newcommand{\bluethree}[0]{{\color{blue}3}}

\newcommand{\greenfour}[0]{{\color{green}4}}

\newcommand{\blank}[0]{\color{red}{  }}
\newcommand{\id}{\mathrm{id}}
\renewcommand{\bar}[1]{\overline{#1}}

\newcommand{\vsubseteq}{{\rotatebox[origin=c]{-90}{$\subseteq$}}}

\newcommand{\thm}[1]{Theorem~\ref{#1}}

\newcommand{\fig}[1]{Figure~\ref{#1}}

\begin{document}


    \pagenumbering{roman}
    \pagestyle{plain}

    %
    %

    \singlespacing

    ~\vspace{-0.75in} 
    \begin{center}

        \begin{huge}
        A quantum algorithm for the quantum Schur-Weyl transform
        \end{huge}\\\n
        By\\\n
        {\sc Sonya J. Berg}\\
        B.S. (UC Santa Barbara) 2003\\
        M.A. (CSU Sacramento) 2005\\\n
        DISSERTATION\\\n
        Submitted in partial satisfaction of the requirements for the degree of\\\n
        DOCTOR OF PHILOSOPHY\\\n
        in\\\n
        MATHEMATICS\\\n
        in the\\\n
        OFFICE OF GRADUATE STUDIES\\\n
        of the\\\n
        UNIVERSITY OF CALIFORNIA\\\n
        DAVIS\\\n\n

        Approved:\\\n\n

        \rule{4in}{1pt}\\
        ~Greg Kuperberg (Chair)\\\n\n

        \rule{4in}{1pt}\\
        ~Bruno Nachtergaele\\\n\n

        \rule{4in}{1pt}\\
        ~Jes\'us De Loera\\

        \vfill

        Committee in Charge\\
        ~2012

    \end{center}

    \newpage

    %
    %




    \doublespacing

    %
    %

    \tableofcontents

    \newpage


    \centerline{\large A quantum algorithm for the quantum Schur-Weyl transform}

    \centerline{\textbf{\underline{Abstract}}}

We construct an efficient quantum algorithm to compute the
quantum Schur-Weyl transform for any value of the quantum parameter
$q \in [0,\infty]$.  Our algorithm is a $q$-deformation of the
Bacon-Chuang-Harrow algorithm \cite{BCH}, in the sense that it
has the same structure and is identically equal when $q=1$.  When $q=0$,
our algorithm is the unitary realization of the Robinson-Schensted-Knuth
(or RSK) algorithm, while when $q=\infty$ it is the dual RSK algorithm
together with phase signs.  Thus, we interpret a well-motivated quantum
algorithm as a generalization of a well-known classical algorithm.

    \newpage

    %
    %

    \chapter*{\vspace{-1.5in}Acknowledgments and Thanks}

I would like to thank my friends and family who supported me during my long
stint in graduate school. In particular, I'd like to thank Chris Berg, who
I met many years ago during my first upper division math course in complex
analysis, and who I've known my entire mathematical career. I've often
wondered whether I would have gone as far in math as I have without his
guidance and support. I also thank him for being a most excellent father
to our daughter Kai, allowing me the time to pursue my own intellectual
interests. I'd also like to thank his parents Patti and Rod Berg for many
hours of babysitting and financial support.

Unlike many mathematicians, my love of math developed at a later age, and
thanks are due to some specific people who cultivated my interest. First
I thank an old friend Brand Belford who first gave me a book when I was
20 years old on the solving of Fermat's last theorem that inspired me to
take my first class on mathematical proofs. I have been lucky to have so
many talented and inspirational teachers along the way. In particular, I'd
like to thank Mihai Putinar, who I'll always remember told me, ``You're a
mathematician; I can see it in your eyes''. I pay the highest respect to
my PhD advisor Greg Kuperberg, who inspired all the work in this thesis
and who has been incredibly patient with me over the many years this took
to complete. Sometimes I'm amazed by how many times he explained the same
thing to me, over and over, with no derision. His perspective on mathematics
is unparalleled and it was truly an honor to learn from him.

I thank Jes\'us De Loera and Bruno Nachtergaele for reading this thesis,
as well as organizing research groups that I attended and enjoyed.

I developed friendships with so many amazing and talented people during my
time in graduate school. In particular, I'd like to thank Hillel Raz, David
and Frances Sivakoff, Owen Lewis, Mohamed Omar, Rohit Thomas, and Gabriel
Amos for inspiring me in all sorts of ways and helping me with babysitting
and emotional support. In particular, I would like to give a special thanks
to the amazing Corrine Kirkbride for always having my back and supporting
me when I've been at my worst. Now that we all live in different cities
I appreciate those years we were together that much more.  I'll always
fondly remember being a mathlete at Sophia's trivia. I'd also like to
thank Joy Jaco Pope for telling me like it is and being a great neighbor.

Thanks go to the University of Toronto Math department for hooking me up
with office space, internet access, and a library card during my year in
their city. I'd also like to thank Karene Chu for working with me day in
and day out, in Huron, and coffee shops across T.O. I thank my friend Jeff
Latosik for constantly expressing his firm belief in my ability to finish
my research when I was doubtful. I'd also like to thank my grandma Pat
Monahan, and aunts and uncles Pat Monahan, John Monahan, Cristina Alvarez,
and Dave Jones, for supporting me financially during the poorest month of
my adult life which occurred during my time in Toronto. Thanks also go to
my mom Helen Monahan.

Lastly I would like to honor the memory of Joshua Gooding. Whenever I had
a rough day with research I would think of how much he would have given
to have the opportunity to complete his thesis. You are missed.

    \newpage

    %
    %

    \pagestyle{fancy}
    \pagenumbering{arabic}

    %
    %

\newchapter{Intro}{Introduction}{Introduction}\label{ch:intro}
\input{intro.tex}

\newchapter{Rep theory}{Hopf algebra representation theory}{Hopf algebra representation theory}
\label{ch:hopf}

\section{Introduction}

Finite groups and semisimple Lie algebras are familiar examples of
algebraic structures with nice representation theories. Hopf algebras
have an algebraic structure which generalizes that of both finite groups
and semisimple Lie algebras, while retaining the key properties of their
representations. In this chapter we describe Hopf algebras and the basics of
their representation theory. In Section \ref{s:hopf} we define Hopf algebras
and the property of cocommutativity. In Section \ref{s:qgroup} we see an
example of a noncocommutative Hopf algebra which will reappear in subsequent
chapters. In Section \ref{s:hreps} we define the representation theory of
Hopf algebras and state some of the key theorems in their study. Finally,
in Section \ref{s:gt} we describe a basis for algebra representations
which is of both algebraic and computational interest.

\section{Hopf algebras}\label{s:hopf}
\input{hopf.tex}

\section{The quantum group of $\mathfrak{gl}(d)$}\label{s:qgroup}
\input{qgroup.tex}

\section{The representation theory of Hopf algebras}\label{s:hreps}
\input{reptheory.tex}

\section{Gelfand-Tsetlin type bases}\label{s:gt}
\input{gtsetlin.tex}

\newchapter{comb}{The combinatorics of Young tableaux and insertion algorithms}{The combinatorics of Young tableaux and insertion algorithms}\label{ch:comb}

\section{Introduction}

The representation theories of the algebras described in this thesis are
indexed by combinatorial objects called Young tableaux. In this chapter
we describe the combinatorics of these objects. In Section \ref{s:comb}
we define Young tableaux and state some of their key properties. In Section
\ref{s:qinsertion} we describe insertion algorithms for operating on Young
tableau, which will connect to some interesting representation theory in
subsequent chapters.

\section{The combinatorics of Young tableaux}\label{s:comb}
\input{combinatorics.tex}

\section{Insertion algorithms}\label{s:qinsertion}
\input{qinsertion.tex}

\newchapter{quantumgroupreptheory}{The representation theories of the quantum
group $U_q(d)$ and the Hecke algebra $H_q(n)$}{The representation theories
of the quantum group $U_q(d)$ and the Hecke algebra $H_q(n)$}\label{ch:qreps}

\section{Introduction}

In this chapter we describe the representation theories of our quantum
algebras of interest. In Section \ref{s:qrep} we present the respresentation
theory of the quantum group seen in Section \ref{s:qgroup}. In Section
\ref{s:heckealgebra} we present the representation theory of the Hecke
algebra $H_q(n)$, which is a $q$-deformation of the symmetric group algebra
$\C[S(n)]$. These Hecke algebras are not themselves Hopf algebras, but for
almost all choices of $q$ their representations are isomorphic to those of
$\C[S(n)]$, which is a Hopf algebra. Finally in Section \ref{s:schurweyl}
we describe the correspondence known as Schur-Weyl duality between
representations of the quantum group and Hecke algebra.

\section{The representation theory of $U_q(d)$}\label{s:qrep}
\input{qreps.tex}

\section{The representation theory of the Hecke algebra $H_q(n)$}\label{s:heckealgebra}
\input{heckealgebra.tex}

\section{Schur-Weyl duality}\label{s:schurweyl}
\input{schurweyl.tex}

\newchapter{qschur}{The Pieri and Schur-Weyl transforms}{The Pieri and Schur-Weyl transforms}\label{ch:qschur}
\section{Introduction}

In Chapter \ref{ch:qreps} we presented the representation theories of
the quantum group $U_q(d)$ and the Hecke algebra $H_q(n)$, and their
correspondence via Schur-Weyl duality. In this chapter we present a
transform which we prove computes a Schur-Weyl transform. In Section
\ref{s:pieri} we define a Pieri transform via only the representation
theory of the quantum group $U_q(d)$. In Section \ref{s:cascade} we compose
Pieri transforms and prove that this computes a Schur-Weyl transform. This
result is stated without proof in \cite{BCH}, and we were unable to find
a proof in the literature. Finally in Section \ref{s:crystal} we look at
the Pieri and Schur-Weyl transforms in their crystal limits, which ties
representation theory together with the insertion algorithms seen in
Section \ref{s:qinsertion}.

\section{The Pieri transform}\label{s:pieri}
\input{pieri.tex}
\section{The Schur-Weyl transform}\label{s:cascade}
\input{cascade.tex}
\section{Pieri and Schur-Weyl in the crystal limit}\label{s:crystal}
\input{crystal.tex}

\newchapter{qalgorithm}{A quantum algorithm for the quantum Schur transform}{A quantum algorithm for the quantum Schur transform}
\label{ch:schurxform}

\section{Introduction}

In this chapter we present our main theorem, which is a quantum algorithm
for computing Schur-Weyl duality.  A quantum computer uses basic units of
information called \emph{qubits} rather than bits. Since the algebra of
qubit spaces is quite different than that of bits, the possible types of
qubit transformations, and thus the maps considered by quantum computers
in their calculations, are different as well. This in turn influences the
types of things that can be calculated within a certain time complexity.

In Section \ref{s:qprob} we describe the basics of quantum probability, which
forms the algebraic base for quantum computation. In Section \ref{s:qalgs}
we review the basic theory of quantum algorithms, and in Section \ref{s:main}
we present our main theorem. The methods we use in the proof of our main
theorem are modeled on those found in the paper \cite{BCH}, where they
prove the existence of a Schur-Weyl algorithm for the case $q = 1$. As
far as we know, this thesis contains the first instance of an algorithm
designed for a quantum computer for the purpose of decomposing quantum
algebra representations.

\section{Quantum probability}\label{s:qprob}
\input{qprob.tex}
\section{Quantum algorithms}\label{s:qalgs}
\input{qalgorithms.tex}
\section{Time complexity of Schur-Weyl}\label{s:main}
\input{schurxform.tex}

\newchapter{Conclusion}{Conclusion}{Conclusion}
\label{ch:conclusion}
\input{conclusion.tex}
    %

\end{document}

%% file: intro.tex
This thesis addresses some problems in quantum computation that are
motivated by quantum algebra.

The quantum Fourier transform for a finite group $G$ plays a central role
in the theory of quantum algorithms.  This is another name for the Burnside
decomposition of the group algebra of $G$,
$$\C[G] \cong \bigoplus_V V \tensor V^*,$$
which is an isomorphism of Hilbert spaces as well as an isomorphism
of algebras.  Since the Burnside decomposition is a Hilbert space
isomorphism and therefore a unitary operator, one can ask when it
can be expressed by a small quantum circuit, or equivalently, when it
has a fast quantum algorithm.

Polynomial-time quantum algorithms for the Burnside decomposition are known
for many finite groups (see for example \cite{shor},\cite{Beals},\cite{MRR}).
In especially favorable cases, the quantum Fourier transform for $G$
yields an algorithm for the hidden subgroup problem for $G$ or other
groups related to $G$.  In particular, the Shor-Simon-Kitaev algorithm
(see \cite{shor}, \cite{simon}, \cite{kitaev}) to find periods or compute
discrete logarithms in any finitely generated abelian group is based on
the quantum Fourier transform for finite abelian groups.

The Schur-Weyl decomposition is another transform which is related to the
Burnside decomposition for the symmetric group $S(n)$.  Given a Hilbert space
or \emph{qudit} $V = \C^d$, the Schur-Weyl decomposition is
\begin{align}\label{sw} V^{\tensor n} \cong \bigoplus_{\lambda \vdash n} R^\lambda \tensor V^\lambda, \end{align}
where $R^\lambda$ is an irreducible representation of the symmetric
group $S(n)$, acting by permuting tensor factors, while $V^\lambda$
is an irreducible representation of the unitary group $U(d)$, which acts
simultaneously or diagonally on all of the factors of $V$.  The fact
that $R^\lambda$ is the multiplicity space of $V^\lambda$ and vice versa
is known as Schur-Weyl duality.  Recently, Bacon, Chuang, and Harrow
presented an efficient quantum algorithm to compute a basis refinement of
this decomposition \cite{BCH}.

In this thesis, we clarify and generalize the Bacon-Chuang-Harrow (or BCH)
algorithm.  First, the Schur-Weyl decomposition has a generalization
that depends on a parameter $q$ from quantum algebra.  We replace the
unitary group $U(d)$ with the quantum group $U_q(d) = U_q(\gl(d))$,
and the symmetric group $S(n)$ with the Hecke algebra $H_q(n)$.  Then the
Schur-Weyl decomposition still exists for every $q \in \C$ which is not a
root of unity, but the specific linear isomorphism expressed by equation
\ref{sw} depends on $q$.  (If $q$ is a root unity of order $r$, then
the decomposition still exists, but it degenerates into a different form
when $r = O(n+k)$.)  When $q$ is real and positive, then both sides of
\ref{sw} are naturally Hilbert spaces and the isomorphism is still
unitary.

Our main result is the following theorem which appears in Section \ref{s:main}.
\begin{theorem}\label{th:main} There is an efficient continuous family of
quantum algorithms for the quantum Schur-Weyl transform for each $q \in
[0,\infty]$. When $q = 1$, the algorithm is the Bacon-Chuang-Harrow
algorithm.  The algorithm continuously extends to $q = 0$ and becomes
a unitary form of the Robinson-Schensted-Knuth (RSK) algorithm
(\cite{schensted}, \cite{K}) together with phase signs.
The algorithm also continuously extends to $q=\infty$, and becomes the
dual RSK algorithm without any phases.\end{theorem}

Note the double use of the word ``quantum", referring to both quantum
computation and quantum algebra.  Those constructions in quantum algebra
that are non-unitary have no quantum computation interpretation, while many
constructions in quantum computation only have a pro forma interpretation
as quantum algebra. Theorem \ref{th:main} properly lies in both topics.  In fact,
the two senses of quantumness are slightly incongruous.  In quantum algebra,
the $q=1$ case is called classical or non-quantum, because it is the case
in which quantum groups become ordinary groups.  But as an algorithm, the
Schur-Weyl transform is not classical when $q=1$; it becomes classical
when $q=0$ instead.  The limit $q=0$ is called the \emph{crystal limit}
in quantum algebra.

Like the BCH algorithm, our algorithm has running time polynomial in the
number of qudits $n$, the size of the qudit $d$, and $\log \eps^{-1}$,
where $\epsilon$ is the desired accuracy.  The bound on running time is
also uniform in $q$, assuming that $q$ itself can be computed quickly.
Therefore, our algorithm is efficient in the sense that it is polynomial in
the number of qudits, for any fixed size of qudit.  We do not know whether
there is an algorithm which is jointly polynomial in $n$ and $\log d$,
\ie, polynomial in the input qubit length $n(\log d)$.

Our algorithm can be compared to quantum straightening algorithms
\cite{LT}.  Our algorithm can be called a Schur-Weyl straightening
algorithm, but we emphasize a different interpretation.  Straightening
algorithms are traditionally interpreted as algorithms in symbolic algebra or
numerical analysis.  As such, the input is not a linear number of qubits or
qudits, but rather an exponential list of components of a vector in a vector
space such as $V^{\tensor n}$.  One can make the same distinction between
a quantum Fourier transform and a classical discrete Fourier transform,
which can be algebraically the same, but are interpreted differently as
computer science.  One interesting connection between the two interpretations
is that a polynomial-time algorithm for a quantum transform always yields
a quasilinear-time algorithm for the corresponding numerical transform.
(The converse does not hold in general.)

Finally, in our interpretation and proof of \ref{th:main}, we will be
more precise about the basis refinement of \ref{sw}.  The relevant
basis of $V^\lambda$, which can be called the ``insertion tableau" by
extension from the $q=0$ and $q=\infty$ cases, is its Gelfand-Tsetlin-Jimbo (or GTJ)
basis.  The BCH algorithm and our generalization compute their
result in this basis essentially by construction --- the algorithm is built
from a subroutine, the Pieri transform (which BCH call the Clebsch-Gordan
transform), that stays in this basis.  We will also prove that the algorithm
yields the Young-Yamanouchi-Hoefsmit (or YYH) basis of $R^\lambda$, up to sign.
(Bacon, Chuang, and Harrow state that it produces the Young-Yamanouchi
basis without proof.)  Finally, specific bases of each $V^\lambda$ and
$R^\lambda$ do not quite completely determine a basis of the right side
of equation \ref{sw}, because we could still multiply each summand $R^\lambda
\tensor V^\lambda$ by a scalar, or in the Hilbert space case, by a phase.
In this sense, the Schur-Weyl transform is not quite uniquely determined.

This thesis is structured as follows.  In Chapter \ref{ch:hopf} we describe
the representation theory of Hopf algebras. In particular, we  focus on
the Hopf algebra of interest in this thesis, the quantum group $U_q(d)$. In
Chapter \ref{ch:hopf} we also investigate the Gelfand Tsetlin type bases for
representations, which have properties desirable for quantum computation. In
Chapter \ref{ch:comb} we describe the necessary combinatorics to discuss the
representation theory defining the Schur-Weyl transform. We also describe
the RSK algorithm and a generalization, which we call quantum insertion. In
Chapter \ref{ch:qreps} we describe the representation theories of the quantum
group $U_q(d)$ and the type A Hecke algebra $H_q(n)$, using the combinatorial
language detailed in Chapter \ref{ch:comb}. We end the chapter with the
formulation of Schur-Weyl duality, which is central to this thesis. In
Chapter \ref{ch:qschur} we define the Pieri and Schur-Weyl transforms with
an emphasis on their connections with insertion algorithms. Finally, in
Chapter \ref{ch:schurxform} we give an introductory backgroung to quantum
probability and algorithms, and present our main theorem \ref{th:main}.

%% file: hopf.tex
The material in this section can be found in \cite{Ka}. An algebra $A$
over a field $\F$ has an associative multiplication $m \colon A \otimes
A \to A$. It also has a two-sided unit, which can be expressed as a map
$\iota \colon \F \to A$ such that $\iota(1)\cdot a = a = a \cdot \iota(1)$
for all $a$ in $A$. In pictures, the following two diagrams should commute:

(Associativity Axiom)
\[\xymatrix{A \otimes A \otimes A \ar[d]^-{\text{id} \otimes m} \ar[r]^-{m \otimes \text{id}} & A \otimes A \ar[d]^-{m} \\
        A \otimes A \ar[r]^-{m} & A}\]

(Unit Axiom)
\[\xymatrix{\F \otimes A \ar[rd]^-{\cong} \ar[r]^-{\iota \otimes \text{id}} & A \otimes A \ar[d]^-{m} & A \otimes \F \ar[l]_-{\text{id} \otimes \iota} \ar[ld]_-{\cong} \\
        & A}\]

A \emph{coalgebra} is obtained by reversing all the arrows. Thus we have
a coassociative comultiplication $\Delta \colon A \to A \otimes A$ and a
two-sided counit $\varepsilon \colon A \to \F$, where the following two
diagrams should commute:

(Coassociativity Axiom)
\[\xymatrix{A \otimes A \otimes A & A \otimes A \ar[l]_-{\Delta \otimes \text{id}} \\
        A \otimes A \ar[u]_-{\text{id} \otimes \Delta} & A \ar[u]_-{\Delta}\ar[l]_-{\Delta}}\]

(Counit Axiom)
\[\xymatrix{\F \otimes A & A \otimes A \ar[l]_-{\varepsilon \otimes \text{id}} \ar[r]^-{\text{id} \otimes {\varepsilon}} & A \otimes \F \\
        & A \ar[lu]_-{\cong} \ar[u]_-{\Delta} \ar[ru]^-{\cong} }\]

If an algebra $A$ also has a coalgebra structure, so that the maps $\Delta,
\varepsilon$ are algebra homomorphisms and $m$, $\iota$ are coalgebra
homomorphisms, then it is a \emph{bialgebra}. Then, a \emph{Hopf algebra}
is a bialgebra with a map called the \emph{antipode}. The antipode is
a bialgebra endomorphism $S \colon A \to A$ where the following three
compositions are identical:

\[\xymatrix{ A \ar[r]^-{\Delta} & A \otimes A \ar[r]^-{S \otimes \text{id}} & A \otimes A \ar[r]^-{m} & A }\]
\[\xymatrix{ A \ar[r]^-{\Delta} & A \otimes A \ar[r]^-{\text{id} \otimes S} & A \otimes A \ar[r]^-{m} & A }\]
\[\xymatrix{ A \ar[r]^-{\varepsilon} & \F \ar[r]^-{\iota} & A }\]

We will generally consider algebras over the complex numbers $\C$ and the
real numbers $\R$.
\begin{example} Given a group $G$, we can form its group algebra $\C[G]$
with basis indexed by elements $g \in G$. $\C[G]$ is in fact a Hopf algebra
with coproduct, counit, and antipode map defined by
\begin{align*}
\Delta(g) & = g \otimes g\\
\varepsilon(g) & = 1\\
S(g) & = g^{-1}
\end{align*}
\end{example}

\begin{example} Let $\mathfrak{g}$ be a Lie algebra over $\C$. Then its
universal enveloping algebra $U(\mathfrak{g})$ is a Hopf algebra with
coproduct, counit, and antipode map defined by
\begin{align*}
\Delta(a) & = a \otimes 1 + 1 \otimes a\\
\varepsilon(a) & = 0\\
S(a) & = -a
\end{align*}
\end{example}

If $A$ is an algebra over $\C$, then it is a *-algebra if it has a map $*
\colon A \to A$ with the following properties:
\begin{align*}
(a+b)^* &= a^*+b^* & (\lambda a)^* &= \overline{\lambda} a^* \\
(ab)^* &= b^*a^* & a^{**} &= a,
\end{align*}
for $a, b \in A$ and $\lambda \in \C$. If $A$ is a *-algebra and Hopf
algebra so that $\Delta(x^*) = \Delta(x)^*$, then we call $A$ a \emph{Hopf
*-algebra}.

If $A_{\R}$ is an algebra over $\R$, then $A_\C = A_\R \otimes_\R \C$
is an algebra over $\C$. On the other hand a complex algebra $A_\C$ may
have more than one decomplexification $A_\R$, even though there is always
an obvious algebra inclusion $A_\R \subseteq A$.

Specifying a decomplexification of $A$ is equivalent to choosing a ``bar structure"
$a \mapsto \bar{a}$ that satisfies the axioms:
\begin{align*}
\bar{a+b} &= \bar{a}+\bar{b} & \bar{\lambda a} &= \bar{\lambda} \bar{a} \\
\bar{ab} &= \bar{a}\bar{b} & \bar{\bar{a}} &= a.
\end{align*}
This is almost the same as a *-structure, the difference being that
a bar structure does not reverse multiplication.  Given a bar structure
on $A$, the real subalgebra $A_\R$ is the set of self-conjugate
elements $a = \bar{a}$.  Also, if $A$ has both a *-structure and
a bar structure, then we require that they commute, or
$$(\bar{a})^* = \bar{(a^*)}.$$

In some cases, such as for $\C[G]$, the antipode map is involutory,
and the *-map is essentially a conjugate-linear version of the antipode
map. In other cases, the antipode map will not be involutory and there is
some other *-map making the algebra into a Hopf *-algebra.

We end by describing the condition of commutativity and define the analogue
for the coalgebraic structure of a Hopf algebra.  Let $A$ be an algebra
with multiplication map $m$. Define the flip map $\tau \colon A \otimes A
\to A \otimes A$ by $\tau(x \otimes y)=y \otimes x$. One way of defining
$A$ to be commutative is by requiring the following diagram commute:

\[\xymatrix{A \otimes A \ar[rd]^-{m} \ar[rr]^-{\tau} & & A \otimes A \ar[ld]_-{m}\\
& A & }\]

Thus, if $A$ is a coalgebra with comultiplication map $\Delta$, we define
\emph{cocommutativity} by requiring the following diagram commute instead:

\[\xymatrix{A \otimes A & & A \otimes A \ar[ll]_-{\tau}\\
       & A \ar[lu]_-{{\Delta}} \ar[ru]^-{{\Delta}} & }\]

Note that both $\C[G]$ and $U(\mathfrak{g})$ are generally noncommutative
but always cocommutative. In the next section we'll examine a Hopf algebra
which is both noncommutative and noncocommutative.

%% file: qgroup.tex
The Hopf algebras $\C[G]$ and $U(\mathfrak{g})$ we saw in section
\ref{s:hopf} are generally noncommutative but always cocommutative. In this
section, we introduce an example of a Hopf algebra which is noncommutative
and noncocommutative: the \emph{quantum group}. Quantum groups as defined
independently by Drinfeld \cite{D} and Jimbo \cite{J} are deformations of
$U(\mathfrak{g})$ for $\mathfrak{g}$ a Lie aglebra.

In this section we consider the Lie algebra $\mathfrak{gl}(d)$ which
is isomorphic to $\text{End}(d)$, the set of linear maps on $\C^d$. The
generators of $U(\mathfrak{gl}(d))$ are $e_i$ and $f_i$ for $1 \leq i \leq
d-1$, and $h_i$ for $1 \leq i \leq d$. The relations on the generators
are called Serre relations and are given by:
\begin{align*}
[h_i,h_j] & = 0 \qquad \text{for $j \ne i$}\\
[h_i,f_j] & = [h_i,e_j] = 0\\
[e_i,f_j] & = \delta_{ij} h_i\\
[e_i,e_j] & = [f_i,f_j] = 0 \qquad \text{for $|i-j|>1$}\\
e_ie_{i \pm 1}e_i & = \frac{1}{2}\left(e_i^2e_{i \pm 1} + e_{i \pm 1}e_i^2 \right) \\
f_if_{i \pm 1}f_i & = \frac{1}{2}\left(f_i^2f_{i \pm 1} + f_{i \pm 1}f_i^2 \right).
\end{align*}

The associated Drinfeld-Jimbo quantum deformation of $U(\mathfrak{gl}(d))$
is called a \emph{quantum group}, and is written $U_q(\mathfrak{gl}(d))$
which we will abbreviate to $U_q(d)$. The parameter $q$ is a complex number
not equal to zero or one. The generators of $U_q(d)$ are $e_i$ and $f_i$
for $1 \leq i \leq d-1$, and $q^{\pm h_i/2}$ for $1 \leq i \leq d$. The
generators $q^{\pm h_i/2}$ can be interpreted as formal exponentials
rather than actual powers of $q$.  The formal notation is meant to imply
that these generators commute with each other and that $q^{-h_i/2}$ is
the reciprocal of $q^{h_i/2}$, using addition in the exponent.  In these
formal exponentials, we also let $k_i = h_i - h_{i+1}$.

We use the notation $[n]$ for the \emph{quantum integer} defined by the formula
$$[n] = \frac{q^n - q^{-n}}{q - q^{-1}} = q^{n-1} + q^{n-3} + \dots + q^{-(n-3)} + q^{-(n-1)}.$$

Extending the notation to operators, we write
$[h_i] = \displaystyle\frac{q^{h_i} - q^{-h_i}}{q - q^{-1}}.$

The relations on the $U_q(d)$ generators are $q$-deformations of the
$U(\mathfrak{g})$ Serre relations, and are given by
\begin{align*}
[q^{h_i/2},q^{h_j/2}] & = 0\,\,, \text{for $i \ne j$}\\
q^{h_i/2}e_j &=
\begin{cases} q^{1/2}e_j q^{h_i/2} & \text{for $i = j$}\\
q^{-1/2}e_j q^{h_i/2} & \text{for $i = j+1$}\\
e_j q^{h_i/2} & \text{otherwise} \end{cases}\\
q^{h_i/2}f_j &=
\begin{cases} q^{-1/2}f_jq^{h_i/2} & \text{for $i = j$}\\
q^{1/2}f_jq^{h_i/2} & \text{for $i = j+1$}\\
f_jq^{h_i/2} & \text{otherwise} \end{cases}\\
[e_i,f_j] &=  \delta_{ij} [h_i]\\
[e_i,e_j] &= [f_i, f_j] = 0,\,\, |i - j| \geq 2\\
e_ie_{i \pm 1}e_i & = \frac{1}{[2]}\left(e_i^2e_{i \pm 1} + e_{i \pm 1}e_i^2 \right) \\
 f_if_{i \pm 1}f_i & = \frac{1}{[2]}\left(f_i^2f_{i \pm 1} + f_{i \pm 1}f_i^2 \right).
\end{align*}

Interestingly, the Hopf algebra structure on $U(\mathfrak{g})$ can also
be deformed so that $U_q(d)$ is a Hopf algebra. For example, the coproduct
map becomes
\begin{align*}
& \Delta(q^{h_i/2})  = q^{h_i/2} \otimes q^{h_i/2}\\
& \Delta(e_i)  = e_i \otimes q^{-k_i/2} + q^{k_i/2} \otimes e_i\\
& \Delta(f_i)  = f_i \otimes q^{-k_i/2} + q^{k_i/2} \otimes f_i.
\end{align*}

There are other deformations of the Hopf algebra structure that result
in different coproduct maps. For example, we could replace $\Delta$ as
defined above by $\tau \circ \Delta$, which is distinct from $\Delta$
by noncocommutativity.

When $q$ is real and positive, $U_q(d)$ has a *-map defined by
$$e_i^* = f_i \qquad f_i^* = e_i \qquad (q^{h_i/2})^* = q^{h_i/2}.$$
This *-map makes $U_q(d)$ into a Hopf *-algebra. When $q$ is real and
positive, $U_q(d)$ also has a bar structure in which all of the generators
are real, and they generate a real Hopf *-algebra $U_q(d)_\R$.

In future sections we'll restrict to the case when $q$ is real and positive
so that we can use the associated * and bar structures.

%% file: reptheory.tex
A \emph{representation} of an algebra $A$ is a vector space $V$ and a linear
map $\rho: A \rightarrow \text{End}(V)$ which preserves the multiplication
and unity, i.e. $\rho(ab) = \rho(a)\rho(b)$ and $\rho(1)=1$. The action
$\rho$ can be implied so that $\rho(a)v$ is written $av$, or in quantum
notation as $a\ket{v}$. In the rest of this section, we fix the assumptions
that our algebra $A$ is a Hopf algebra, and our representations $V$ are
defined over $\C$ and are finite dimensional.

Two representations $V$ and $W$ of $A$ are \emph{isomorphic} if there
exists a linear bijection $T: V \rightarrow W$ that commutes with the
action of $A$, \ie, $T(av) = aT(v)$ for all $a \in A$, $v \in V$.

The representation $V$ is \emph{irreducible} if it has no non-trivial
subspaces that are closed under the action of $A$. In this thesis we use the
abbreviation \emph{irrep}. For example, the counit map $\varepsilon \colon
A \to \C$ defines a \emph{trivial representation}, which is irreducible
since it's one dimensional.

Given two representations $V$ and $W$ of $A$, there is a well-defined
representation structure on the direct sum $V \oplus W$ given by
$$a(v \oplus w) = av \oplus aw.$$

A representation $V$ is \emph{semisimple} if it is isomorphic to a direct
sum of irreps. (Likewise, an algebra $A$ is called semisimple if all of its
representations are semisimple.) The number of occurrences of an irrep $W$
in $V$ is called the \emph{multiplicity} of $W$. If the multiplicities
are all 0 or 1, then $V$ is called \emph{multiplicity-free}.

Assume for the moment that $A$ and its subalgebras are semisimple. Given
a representation $V$ of $A$ and a subalgebra $B \subseteq A$, the
restriction of $V$ to $B$ will be denoted by $\Res^{A}_{B} V$. When $V$
is an irrep of $A$, $\Res^{A}_{B} V$ is typically not an irrep of $B$,
but by semisimplicity $\Res^{A}_{B} V$ decomposes as a direct sum of
irreps of $B$.  A rule for describing the decomposition of $\Res^{A}_{B}
V$ into irreps is called a \emph{branching rule}.  If for all irreps $V$
of $A$, the branching rule for $\Res^{A}_{B} V$ is multiplicity-free,
then the inclusion $B \subseteq A$ is called a \emph{Gelfand pair}.

Given two representations $V$ and $W$ of $A$, the coproduct map $\Delta
\colon A \to A \otimes A$ is used to define a representation structure on
$V \otimes W$.

The antipode map is used to define a \emph{dual representation}. Given
a representation $V$ of $A$, define $V^*$ to be the dual space of linear
functionals on $V$. Then, the action of $A$ on $V^*$ is defined by $a\bra{v}
= \bra{v} S(a)$.

If $A$ is a Hopf $^*$-algebra, then $V$ is a *-representation if $\rho(a^*)
= \rho(a)^*$ where the * on the right side is the Hermitian adjoint. (The
Hermitian adjoint makes the algebra $\text{End}(V)$ into a *-algebra.) This
generalizes the notion of a unitary representation of a group. In
particular, a *-representation $V$ is automatically semisimple: If $W$ is
a subrepresentation of $V$, then so is its orthogonal complement $W^\perp$.

Although such a $V$ might possibly have non-orthogonal irreducible
decompositions, it always has an orthogonal irreducible decomposition.
If $V$ is multiplicity-free, then its irreducible decomposition is unique
and therefore orthogonal.

Our analysis so far carries over verbatim to representations of algebras
over $\R$. Quantum computation is defined over $\C$, and we will ultimately
be interested in connecting representations over $\R$ with representations
over $\C$.

If $V_\R$ is a representation of $A_\R$, then
$$V_\C = V_\R \tensor_\R \C$$
is a bar representation of $A_\C$. But note that even if $V_\R$ is
irreducible, $V_\C$ may or may not be irreducible.  If $\End(V_\R) \cong \R$,
then $V_\C$ is irreducible, while if $\End(V_\R) \cong \C$ or $\End(V_\R)
\cong H$ (the quaternions), then $V_\C$ has two irreducible summands.
In the former case, we will say that $V_\R$ is \emph{strongly irreducible}.

\begin{example} $\C[G]$ is a Hopf *-algebra with *-map defined by $g^* =
g^{-1}$. Note that in this case a representation being a *-representation
is the same thing as it being unitary as a representation of $G$. Also
$\C[G]$ has a standard bar structure with $g = \bar{g}$, so that its
decomplexification is the real group algebra $\R[G]$.

When $\C[G]$ is finite dimensional it has additional properties for its
irreps. For example, there are finitely many distinct irreps of $\C[G]$,
indexed by the conjugacy classes of $G$. And we always have semisimplicity
of representations of $\C[G]$.
\end{example}

\begin{example} Every continuous representation $V$ of a connected Lie
group $G$ is also a representation of the universal enveloping algebra
$U(\mathfrak{g})$ and it has the same subrepresentations.

If $\mathfrak{g}_\R$ is a real Lie algebra and $\mathfrak{g}_\C$ is
its complexification, then $U(\mathfrak{g}_\C)$ has both a natural bar
structure --- where the real subalgebra is $U(\mathfrak{g}_\R)$ --- and a
natural *-structure.  Since $U(\mathfrak{g}_\C)$ is generated as a complex
algebra by $\mathfrak{g}_\R$, we define these structures by letting
$$\bar{a} = a \qquad a^* = -a$$
for $a \in \mathfrak{g}_\R$.
\end{example}

%% file: gtsetlin.tex
In this section, we describe bases for irreps with special algebraic and
computational properties.  We will be interested in a tower of algebras
$$\C = A_0 \subseteq A_1 \subseteq A_2 \subseteq \cdots \subseteq A_n$$
and we will use the abbreviation
$$\Res^k_{k-1} V = \Res^{A_k}_{A_{k-1}} V$$
for the restriction of a representation $V$ of $A_k$.

Suppose that each inclusion $A_{k-1} \subseteq A_n$
is a Gelfand pair.  Then if $V = V_n$ is an irrep of $A_n$, $\Res^n_{n-1}
V$ is a direct sum of irreps $V_{n-1}$ of $A_{n-1}$, and by induction
each $\Res^k_{k-1} V_k$ is a direct sum of irreps $V_{k-1}$ of $A_{k-1}$.
As a result, $V$ is expressed as a direct sum of irreps $V_0$ of $A_0
= \C$, and all such irreps are isomorphic and 1-dimensional.  Thus $V$
has a basis of lines which are encoded by flags
$$\C \cong V_0 \subseteq V_1 \subseteq \dots \subseteq V_n = V.$$
This line basis is called a Gelfand Tsetlin type (GTT) basis.  By extension, any vector
basis that refines the GTT line basis is also called a GTT basis.  Note that
in the encoding, the number of bits a GTT basis vector requires is the
sum of the bits required to encode each summand $V_k$.

To get a sense of the significance of a GTT basis, note that whenever
$a \in A_k$ and $v \in V_k$, then $av \in V_k$ also.  This means that
we can express the action of an element $a$ on $V$ in the setting of
a lower-dimensional algebra, which naturally gives rise to a recursive
structure.  However, note that a GTT vector basis of an irrep $V$ is not
unique; only the corresponding line basis is unique.  The
computational strength of a GTT basis can still depend on how its
vectors are scaled.

\begin{remark} In some articles in quantum computation, if $V$ is an
irrep of a group $G$ and $H \subseteq G$ is a subgroup, then a basis that
refines a decomposition of $\Res^G_H V$ is called \emph{subgroup-adapted}.
The analogous notion for us is a basis that is \emph{subalgebra-adapted}.
In this terminology, a GTT basis is recursively adapted to a tower of
subgroups or subalgebras.
\end{remark}

If each algebra $A_k$ is a $^*$-algebra and $V$ is a $^*$-representation of
$A = A_n$, then a GTT basis is automatically orthogonal, because each
restriction $\Res^k_{k-1} V_k$ has an orthogonal decomposition.
We further require that a GTT vector basis of a $^*$-representation be
orthonormal, so that the basis is usable in quantum computation.  However,
even when GTT basis vectors are orthonormal, their phases are still not
determined by the GTT property.

If $V$ and $W$ are two irreps of an algebra $A$, with given GTT bases, then technically
their combinations $V \oplus W$ and $V \otimes W$
do not have GTT bases.  However, we can define standard bases by taking
the direct sum and tensor bases, respectively.  These combinations are
GTT bases with respect to the action of $A \otimes A$ instead.

%% file: combinatorics.tex
A \emph{partition} $\lambda$ is a list of non-negative integers
$$\lambda = (\lambda_1, \lambda_2, \dots, \lambda_d)$$
such that
$$\lambda_1 \geq \lambda_2 \geq \dots \geq \lambda_d.$$

We say that $\lambda$ is a partition of $n$, or $\lambda \vdash n$, if
$\sum_k \lambda_k = n$.  The length of $\lambda$, denoted $\ell(\lambda)$,
is the number of non-zero entries of $\lambda$.

A partition $\lambda$ has an associated \emph{Young diagram}, which is
a horizontal histogram with $\ell(\lambda)$ rows; the $k$th
row has $\lambda_k$ boxes.
\begin{example}
The Young diagram of $\lambda = (3,2,1,1)$ is
$$\yng(3,2,1,1)$$
\end{example}

If $\mu$ and $\lambda$ are partitions so that the Young diagram of $\mu$ is
contained in the Young diagram of $\lambda$, then we write $\mu \subseteq
\lambda$. If $\lambda$ and $\mu$ differ by a single box then $\lambda$
is said to \emph{cover} $\mu$.

When $\mu \subseteq \lambda$ we can form a Young diagram of \emph{skew
shape} given by $\lambda \setminus \mu$ which means removing the boxes in
the Young diagram of $\mu$ from the boxes in the Young diagram of $\lambda$.

\begin{example} If $\lambda = (3,2,1,1)$, and $\mu = (1,1)$, then the skew
shape $\lambda \setminus \mu$ is given by
$$\young(:\blank\blank,:\blank,\blank,\blank)$$
\end{example}

If $\lambda \setminus \mu$ has at most one box in each of its columns,
then it is called a \emph{horizontal strip}.

A \emph{Young tableau} of shape $\lambda$ (including skew shapes) is a
filling of the boxes of the Young diagram of shape $\lambda$ with positive
integers.  If $t$ is a Young tableau of shape $\lambda$, we write $\sh(t)
= \lambda$.  We will use special types of tableaux called semi-standard
and standard Young tableaux.

A Young tableau of shape $\lambda$ is \emph{semi-standard} (abbreviated
SSYT) if its entries weakly increase from left to right and strictly
increase from top to bottom.

\begin{example}
An example of an SSYT with shape $(3,2,1,1)$ is given by
$$t = \young(112,23,3,4).$$

An example of an SSYT with skew-shape $(3,2,1,1) \setminus (1,1)$ is given by
$$u = \young(:11,:3,2,3).$$
\end{example}

We denote the set of SSYT of shape $\lambda$ with entries in $\{1,\dots,d\}$
by $\SSYT(\lambda,d)$.  When the value of $d$ is obvious, we suppress it
and write $\SSYT(\lambda)$.

A Young tableau of shape $\lambda \vdash n$ is \emph{standard} (abbreviated
SYT) if its entries are in $\{1,\dots,n\}$ and strictly increase both from
top to bottom and from left to right.

\begin{example}
An example of an SYT with shape $(3,2,1,1)$ is given by
$$t = \young(124,35,6,7).$$

An example of an SYT with skew-shape $(3,2,1,1) \setminus (1,1)$ is given by
$$u = \young(:12,:4,3,5).$$
\end{example}

We denote the set of standard Young tableaux of shape $\lambda$ by
$\SYT(\lambda)$.

If $\nu$ is a horizontal strip, then the SYT $u$ obtained by filling
the Young diagram of $\nu$ with letters from left-to-right is called
\emph{ordered}. (This is not in general defined for skew-tableau, but
exists for horizontal strips.)

\begin{example}
The ordered SYT of the horizontal strip $\nu = (3,1) \setminus (1)$ is given by
$$u = \young(:23,1).$$
\end{example}

Given any $t \in \SSYT(\lambda,d)$, let $t^{(k)}$ be the restricted tableau in
$\SSYT(\lambda,k)$ obtained by removing all boxes from $t$ with numbers
larger than $k$.

\begin{example} Let
$$t = \young(112,3).$$
Then
$$t^{(2)} = \young(112) \qquad t^{(1)} = \young(11).$$
\end{example}

Note that the skew shapes $sh(t^{(k)}) \setminus sh(t^{(k-1)})$ for an
SSYT $t$ are always horizontal strips for each $k$, so we give them the
label $\lambda^{(i)}(t)$ where $sh(t) = \lambda$.

Finally, the \emph{residue} of a box $b$ in a Young tableau $t$ is the
difference of its coordinates.  In other words, if $b$ has coordinates
$(i,j)$, then its residue is $res(b) = j - i$. Two boxes in a Young tableau
have the same residue if and only if they lie on the same diagonal.

The \emph{axial distance} $a$ between two boxes $b$ and $b'$ is defined to
be the difference in their residues, given by $a = res(b) - res(b')$. It's
described as a distance because it counts the number of boxes in any path
in the Young diagram from the box $b$ to the box $b'$ where moves left
and down count for $+1$ and moves right and up count for $-1$.

\begin{example}
Consider the SYT
$$t = \young(12\redthree,\redfour5,6).$$
The axial distance from the box containing three to the box containing
four is 3. Note the axial distance is antisymmetric, so the distance from
the box containing four to the box containing three is $-3$.
\end{example}

Given a horizontal strip $\lambda$, we define $a_{ij}$ for $i < j$ to be the
axial distance from the last box in the $i$th row to the last box in the
$j$th row. (Note we can equivalently use the boxes in the second-to-last,
third-to-last, etc., positions, and get the same values.)

We will see the axial distances $a_{ij}$ in subsequent chapters when we
describe matrix coefficients that derive from representation theory. One
main reason these axial distances are chosen is that they sum in a very
natural way:
$$a_{ij} + a_{jk} = a_{ik}.$$

\begin{example}
If $\nu$ is the shape given by

$$\young(::::\blank\blank,:\blank\blank\blank,\blank)$$
then $a_{12} = 3$, $a_{23} = 4$, and $a_{13} = a_{12} + a_{23} = 7$.
\end{example}

%% file: qinsertion.tex
Given an SSYT $t$ and a letter $i$, we will add a new box to $sh(t)$ to
make space for an extra letter and insert the $i$ into the tableau $t$,
possible rearranging other letters in the process. In this section we review
the well-known RSK insertion algorithm and introduce a generalization that
we call quantum insertion. The RSK algorithm can be found in \cite{K}
and \cite{schensted}, while quantum insertion is our way of describing
the techniques found in \cite{Da}.

\subsection{RSK insertion}

The first insertion algorithm we examine is called
\emph{Robinson-Schensted-Knuth} (abbreviated RSK), and denoted $(i
\xrightarrow{\text{\tiny RSK}} t)$. This insertion algorithm produces a
unique tableau, given by the rules:

\begin{enumerate}
\item If $i$ is greater than or equal to all the numbers in the first row
of $t$, then add $i$ to the end of the first row of $t$.
\item Otherwise, pick the leftmost box in the first row containing a number
$j > i$.  Replace $j$ by $i$. (This process is referred to as $i$ \emph{bumping} $j$.)
\item Repeat steps (1) and (2) for $j$ starting with the second row. Proceed
inductively.
\end{enumerate}

\begin{example}
Start with
$$t = \young(11\bluetwo,2\redthree,3,4)$$
If we choose to insert a letter $i \ge 2$, such as $i = 4$, then it will
be added to the end of the first row:
$$(4 \xrightarrow{\text{\tiny RSK}} t) = \young(1124,23,3,4)$$
However, if we choose to insert the letter $i = 1$, it will bump the two
out of the first row:
$$\young(111,2\redthree,3,4)$$
The two will then bump the three out of the second row, which will itself
get added at the end of the third row. Therefore,
$$(1 \xrightarrow{\text{\tiny RSK}} t) = \young(111,2\bluetwo,3\redthree,4)$$
\end{example}

There is also a dual RSK algorithm, denoted $(i \xrightarrow{\text{\tiny
RSK}^*} t)$, which can be thought of as the standard RSK algorithm applied
to columns instead of rows. Thus, the dual RSK algorithm also produces a
unique output, given by the rules:
\begin{enumerate}
\item If $i$ is larger than all numbers in the first column of $t$, add $i$
to the end of the first column of $t$.
\item Otherwise, pick the topmost box in the first column that contains a
number $j \geq i$.  Replace $j$ by $i$.
\item Repeat steps (1) and (2) for $j$ starting with the second column. Proceed
inductively.
\end{enumerate}

Given a word $w = w_1 \dots w_n$, we can extend the algorithm by induction to define
\begin{align*}
 w \xrightarrow{\text{\tiny RSK}}  & = w_n \xrightarrow{\text{\tiny RSK}} (w_{n-1} \xrightarrow{\text{\tiny RSK}} ( \dots (w_2 \xrightarrow{\text{\tiny RSK}} w_1)\dots))
\end{align*}

\begin{example}
If $w = \redtwo1\bluetwo$,
\begin{align*}
w \xrightarrow{\text{\tiny RSK}} = \young(1\bluetwo,\redtwo)
\end{align*}

Given the word $w' = \redtwo\bluetwo1$ we still obtain the same output as
$w$, i.e. $w \xrightarrow{\text{\tiny RSK}} = w' \xrightarrow{\text{\tiny
RSK}}$. On the other hand, $w'' = 1\redtwo\bluetwo$ results in a different
tableau

\begin{align*}
w'' \xrightarrow{\text{\tiny RSK}}  = \young(1\redtwo\bluetwo)
\end{align*}
\end{example}

The above example shows that the RSK map is not invertible. However,
note there is a way of distinguishing $w \xrightarrow{\text{\tiny RSK}}$
and $w' \xrightarrow{\text{\tiny RSK}}$ by a \emph{recording tableau}
which tracks the order new boxes are added in the sequence of insertions,
as is done in the following example. Then we define $\text{\tiny RSK}
(w) = P(w) \times Q(w)$, where $P(w) = w \xrightarrow{\text{\tiny RSK}}$
and $Q(w)$ is the recording tableau.

\begin{example}
If $w = \redtwo1\bluetwo$,
\begin{align*}
\text{\tiny RSK} (w) = \young(1\bluetwo,\redtwo) \times \young(13,2)
\end{align*}
whereas if
$w' = \redtwo\bluetwo1$,
\begin{align*}
\text{\tiny RSK} (w') = \young(1\bluetwo,\redtwo) \times \young(12,3)
\end{align*}
so that $\text{\tiny RSK} (w) \ne \text{\tiny RSK} (w')$.
\end{example}

The proof of the following theorem that the RSK map is a bijection can be
found in \cite{K}.
\begin{theorem}\label{t:rsk}
Let $V_d^n$ be the set of words in $d$ letters of length $n$. Then, the RSK
map $\text{\tiny RSK}(w) = P(w) \times Q(w)$ is a bijection between $V_d^n$
and the disjoint union $$\displaystyle\coprod_{\substack{\lambda \vdash n\\
\ell(\lambda) \le d}} \text{SSYT}(\lambda) \times \text{SYT}(\lambda).$$
\end{theorem}

\subsection{Quantum insertion}

In this subsection we consider a generalization of RSK we call \emph{quantum
insertion}, or q-insertion, and denoted $(i \xrightarrow{\text{\tiny
qINS}} t)$. Given an SSYT $t$ and a letter $i$, q-insertion produces a
set of output tableaux, one of which is $(i \xrightarrow{\text{\tiny RSK}}
t)$. The rules for constructing the output tableaux are given by:

\begin{enumerate}
\item In all possible ways add a new box to $sh(t)$.
\item In all possible ways, take either of the following two steps.
\begin{itemize}
\item Insert $i$ into the new box. If this step is taken, the algorithm terminates.
\item For any letter $j > i$, $i$ can replace (or bump) $j$. In this case
step 2 is repeated inductively with $j$.
\end{itemize}
\end{enumerate}

\begin{example}
Let $$t = \young(112,2\bluethree,\greenfour),$$
and suppose we wish to insert a 2 after adding a box to the second row.

Then, starting with
\begin{align*} \young(112,2\bluethree\blank,\greenfour), \end{align*}
we insert a \redtwo.  It replaces the \bluethree \, because it can't take
over the new box, and it can't replace the \greenfour.
                \begin{align*} \young(112,2\redtwo\blank,\greenfour) \end{align*}
We then repeat the procedure with the \bluethree, which can either take
over the new box or replace the \greenfour.
\begin{align*}\left(2 \xrightarrow{\text{\tiny qINS}} \young(112,23\blank,4) \right)
= \left\{\young(112,2\redtwo\bluethree,4)\,, \young(112,2\redtwo\greenfour,\bluethree) \right\} \end{align*}
\end{example}
We define a \emph{bumping sign} for an output tableau as follows. For each
letter involved in the bumping procedure, multiply the bumping sign by a
$-1$ if the letter moves to a lower row in the tableau. Note that for RSK,
the bumping sign can be $\pm 1$ whereas for dual RSK the bumping sign is
always +1.

Analogous to the RSK map, given a word $w = w_1 \dots w_n$, we extend the
q-insertion algorithm by induction to define
\begin{align*}
 w \xrightarrow{\text{\tiny qINS}} & = w_n \xrightarrow{\text{\tiny qINS}} (w_{n-1} \xrightarrow{\text{\tiny qINS}} ( \dots (w_2 \xrightarrow{\text{\tiny qINS}} w_1)\dots)).
\end{align*}

Unlike the RSK algorithm, if $w'$ is a permutation of $w$, then the sets $w
\xrightarrow{\text{\tiny qINS}}$ and $w' \xrightarrow{\text{\tiny qINS}}$
are equal. Also unlike RSK, the output of $w \xrightarrow{\text{\tiny qINS}}$
is an entire set of SSYT, and sometimes there is more than one insertion
path in $w \xrightarrow{\text{\tiny qINS}}$ which produces an SSYT $t$.

As we see in the following example, we can distinguish outputs by attaching a
recording tableau which tracks the order in which new boxes are added during
the insertion process. Then we define $\text{\tiny qINS} (w) = \{P_q(w)
\times Q_q(w)\}$, where $P_q(w)$ is a SSYT in $w \xrightarrow{\text{\tiny
qINS}}$ and $Q_q(w)$ is the associated recording tableau.

\begin{example}
Letting $w = 1\redtwo\bluetwo$,
\begin{align*}
w \xrightarrow{\text{\tiny qINS}} = \left\{ \young(1\redtwo\bluetwo)\, ,
\young(1\redtwo,\bluetwo)\, , \young(1\bluetwo,\redtwo) \right\}
\end{align*}
So, the second and third tableaux in $w \xrightarrow{\text{\tiny qINS}}$
are equal. However, we can distinguish the output tableaux by attaching
a recording tableau to each:
\begin{align*}
\text{\tiny qINS} (w) = \left\{ \young(1\redtwo\bluetwo) \times \young(123)\, ,
\young(1\redtwo,\bluetwo) \times \young(12,3)\, ,
\young(1\bluetwo,\redtwo) \times \young(13,2) \right\}
\end{align*}
\end{example}

Thus far, the reason for using the word ``quantum'' in the context of a
combinatorial insertion algorithm is unclear. In the rest of this section
we describe the reason for this choice. Much of the material can be found
in \cite{Da}.

Define the \emph{weighted q-insertion} map by
$$\text{\tiny qINS} (w) = \sum c_{P,Q} P_q(w) \times Q_q(w),$$
for a choice of nonzero constants $c_{P,Q} \in \C[q,q^{-1}]$.

The choice of coefficients $c_{P,Q}$ that interest us is determined by
representation theory and will be described in chapter \ref{ch:qschur}. The
connection between RSK and q-insertion becomes clear in the following
theorem, which can also be found in \cite{Da}.

\begin{theorem}\label{t:qins}
Let $V_d$ be the vector space over $\C[q,q^{-1}]$ with basis $\{1,\dots,d\}$,
and consider the vector space $V_d^{\otimes n}$ of words of length $n$. Let
$V^\lambda$ and $R^\lambda$ be the vector spaces over $\C[q,q^{-1}]$
with bases $\text{SSYT}(\lambda)$ and $\text{SYT}(\lambda)$, respectively.

Then, there exists a choice of coefficients $c_{P,Q}$ so that the weighted
q-insertion map $\text{\tiny qINS} (w) = \sum c_{P,Q} P_q(w) \times Q_q(w)$
defines a vector space isomorphism
$$V_d^{\otimes n} \cong \displaystyle\bigoplus_{\substack{\lambda \vdash n\\ \ell(\lambda) \le d}} V^\lambda \otimes R^\lambda.$$
\end{theorem}

%% file: qreps.tex
We defined the quantum group $U_q(d)$ in Section \ref{s:qgroup} as
an interesting example of a noncommutative and noncocommutative Hopf
algebra. In this section we describe its representation theory.

Recall that we restrict the values of $q$ to real and positive in order to
make use of the star and bar structures available in this case. For these
values of $q$, the irreducible representations of $U_q(d)$ are isomorphic to
those of the unitary group $U(d)$. (This is true for other values of $q$ as
well, namely those values of $q$ which are not roots of unity or zero.) The
irreps of $U(d)$, and hence the irreps of $U_q(d)$, are in bijection with
partitions $\lambda$ whose length is bounded by $d$, written $\ell(\lambda)
\le d$. For the representation indexed by $\lambda$ we write $V^\lambda$.

Restricting to all the generators except $e_{d-1}$, $f_{d-1}$, and
$q^{h_d}$, we realize a copy of $U_q(d-1)$ inside $U_q(d)$. The branching
rule associated to this pairing is given in the following theorem.

\begin{theorem}\label{th:branch}
The algebras $U_q(d)$ and $U_q(d-1)$ form a Gelfand pair. In particular,
if $V^\lambda$ is an irrep of $U_q(d)$, then
\begin{align*} \Res^n_{n-1} V^\lambda = \bigoplus_{\substack{\lambda \setminus \mu\,\text{horizontal strip}\\ \ell(\mu) \leq d-1} } V^\mu. \end{align*}
\end{theorem}

The branching rule (\ref{th:branch}) implies that a GTT basis for the irrep
$V^\lambda$ can be written $\ket{v_t}$ with $t \in SSYT(\lambda)$. Then
$V^\lambda$ is the span of the elements $\ket{v_t}$ so that $\langle v_t
\mid v_s \rangle = \delta_{t,s}$.

The specific GTT basis we use is called Gelfand-Tsetlin-Jimbo (GTJ). The
formulas described in the rest of this section can be found in \cite{Da}. The
action of the generator $q^{h_i}$ on the GTJ basis is the easiest to
describe and is given in the following theorem.

\begin{theorem}\label{th:gtj1}
Let $\ket{v_t}$ be a GTJ basis element indexed by the SSYT $t$. Then,
\begin{align*} q^{h_i/2} \ket{v_t} = q^{x_i(t)/2} \ket{v_t}, \end{align*}
where $x_i(t)$ counts the number of $i$'s in $t$.
\end{theorem}

The generator $f_i$ acts on $\ket{v_t}$ by turning an instance of $i$ in the
tableau $t$ into an $i+1$ (in all possible ways, i.e. in superposition). In
other words, letting $\ket{v_{t_k}}$ be the vector indexed by tableau $t_k$
where the last $i$ in row $k$ is changed into an $i+1$ but is otherwise
identical to $t$, or zero if this is not possible, then
\begin{align*}f_i \ket{v_t} = \sum_k \langle v_{t_k} \mid f_i \mid v_t\rangle
\ket{v_{t_k}}\end{align*}
for some choice of coefficients $\langle v_{t_k} \mid f_i \mid v_t\rangle$,
which we call the \emph{GTJ statistic}. The action of $e_i$ on $\ket{v_t}$
is also defined with GTJ statistics using the relation $e_i^* = f_i$. The
GTJ coefficients are complicated notationally, but in principle derive
from simple combinatorial properties, in particular axial distances,
of the SSYT $t$. Recall we defined the axial distance $a_{ij}$ of $t$,
and associated horizontal strips $\lambda^{(i)}(t)$ in section \ref{s:comb}.

\begin{theorem}\label{th:gtj2}
Let $\ket{v_t}$ be a GTT basis element of $V^\lambda$ indexed by the SSYT
$t$. Then the GTJ statistic is given by
\begin{align}\label{actionformula}
\langle v_{t_k} \mid f_i \mid v_t\rangle = \sqrt{[\lambda^{(i)}_k] [\lambda^{(i+1)}_k + 1]
\prod_{\substack{{j=1}\\ j \ne k}}^{i+1} \frac{[a_{jk} - \lambda^{(i)}_k]}{[a_{jk}]}
\frac{[a_{jk}+\lambda^{(i+1)}_k + 1]}{[a_{jk}+1]}} \end{align}
\end{theorem}

%% file: heckealgebra.tex
The Hecke algebra $H_q(n)$ is a certain $q$-deformation of the group
algebra $\C[S(n)]$.  (More precisely, we consider a Iwahori-Hecke algebra
of type A.  There are also other kinds of Hecke algebras.) Note that the
Hecke algebra $H_q(n)$ for $q \ne 1$ is \emph{not} a Hopf algebra, so the
results in this section are proved independently of the theorems for Hopf
algebra representation theory.

The Hecke algebra $H_q(n)$ with complex parameter $q$ has generators $\{T_1,
\dots, T_{n-1}\}$ with relations
\begin{align*}
& T_iT_j = T_jT_i & \text{for $|i -j| > 1$}\\
& T_iT_{i+1}T_i = T_{i+1}T_iT_{i+1}  & \\
& (T_i-q^{-1})(T_i+q) = 0. &
\end{align*}
The first two relations are known as the \emph{braid} relations and the
third is the \emph{quadratic} relation. (We use the generators used by
Jimbo \cite{J}; the generators due to Iwahori are slightly different.)
When $q = 1$, the third relation simplifies to $T_i^2 = 1$, so that in this
case $T_i$ represents the transposition $s_i = (i,i+1)$ in the symmetric
group. In other words, $H_1(n) = \C[S(n)]$.

As with the quantum group $U_q(d)$, we restrict to the case when $q$
is real and positive. In this case, $H_q(n)$ has both a *-structure and
a bar structure, defined by
$$T_i^* = T_i \qquad \bar{T_i} = T_i.$$

For these values of $q$, the irreducible representations of $H_q(n)$
are isomorphic to those of $\C[S(n)]$. (This is true for other values
of $q$ as well, namely those values of $q$ which are not roots of unity
or zero.) The irreps of $\C[S(n)]$, and hence the irreps of $H_q(n)$
are in bijection with the conjugacy classes of $S(n)$, so are indexed by
partitions of $n$. For the representation indexed by $\lambda \vdash n$
we write $R^\lambda$. Importantly, for these values of $q$, $H_q(n)$
representations remain semisimple.

It is known that the dimension of $R^\lambda$ equals the number of
standard Young tableaux of shape $\lambda$ (given by, for example, the
hook length formula). Therefore, there is a basis of $R^\lambda$ indexed
by $\text{SYT}(\lambda)$. We describe below how GTT bases are naturally
described by $\text{SYT}(\lambda)$.

The Hecke algebra $H_q(n)$ contains many copies of $H_q(n-1)$; we consider
the one obtained by restricting to the generators $T_1,\dots, T_{n-2}$. Thus,
we can describe the restriction of $R^\lambda$ to $H_q(n)$. The corresponding
branching rule is multiplicity-free and has a nice combinatorial description
in terms of the covering relation of Young diagrams.

\begin{theorem}\label{th:heckebranch} The algebras $H_q(n)$ and $H_q(n-1)$
form a Gelfand pair. In particular, if $R^\lambda$ is an irrep of $H_q(n)$,
then
\begin{align}\label{branch} \Res^n_{n-1} R^\lambda = \bigoplus_{\text{$\lambda$ covers $\mu$}} R^\mu. \end{align}
\end{theorem}

Theorem \ref{th:heckebranch} implies a GTT line basis with elements indexed
by sequences of partitions pairwise differing by a single box, i.e. standard
Young tableaux. The vector basis of $R^\lambda$ we use in this thesis that is
a refinement of the GTT line basis defined by Theorem \ref{th:heckebranch}
we call the Young-Yamanouchi-Hoefsmit (YYH) basis. We write YYH basis
elements as $\ket{r_t}$ where $t \in \text{SYT}(\lambda)$. Then $R^\lambda$
is the span of the elements $\ket{r_t}$ with $\langle r_t \mid r_s \rangle
= \delta_{t,s}$.

Define the following action of $S(n)$ on the basis element $r_t$:
\begin{itemize}
\item If $i$ and $i+1$ are in the same row or column of $t$ then $r_{s_i \cdot t}=0$.
\item Otherwise, $r_{s_i \cdot t}= r_{t'}$ where $t'$ is the standard
tableau obtained by switching $i$ and $i+1$ in $t$.
\end{itemize}
The action of $H_q(n)$ on the YYH basis defined below in Formula
\ref{gtformula} is a normalized version of that given in \cite{R}, building
on that found in \cite{h}.
\begin{theorem}  Let $a$ be the axial distance in $t$ from the box containing
$i$ to the box containing $i+1$. The action of $H_q(n)$ on $R^\lambda$
with the YYH basis is defined by
\begin{align}\label{gtformula}
T_i \ket{r_t} = \frac{q^{-a}}{[a]} \ket{r_t} + \sqrt{1 - \frac{1}{[a]^2}} \ket{r_{s_i\cdot t}}
\end{align}
\end{theorem}
\begin{example}
Consider $R^{(2,1)}$ with basis $\young(12,3),\young(13,2)$. Then,
\begin{align*} T_1  = \left(\begin{array}{cc} q^{-1} & 0\\ 0 & -q \end{array}\right), \qquad\qquad
T_2  = \left(\begin{array}{cc} \frac{q^{-2}}{[2]} & \frac{\sqrt{[3]}}{[2]}\\ \frac{\sqrt{[3]}}{[2]} & \frac{-q^2}{[2]} \end{array}\right).\\
\end{align*}
\end{example}

%% file: schurweyl.tex
Let $V$ be any finite-dimensional vector space over $\C$. Then $V^{\otimes
n}$ is a representation of $\C[S(n)]$ via the simple permutation action
\begin{align}\label{perm}
\pi(v_1 \dots v_n) = v_{\pi^{-1}(1)} \otimes \dots \otimes v_{\pi^{-1}(n)}.\end{align}
The vector space $V^{\otimes n}$ is also a representation of the Hecke
algebra $H_q(n)$ via a $q$-deformation of the permutation action defined
in \ref{perm}. In particular, the generator $T_i$ acts on a vector in
$V^{\tensor n}$ by the identity on all factors except the $i$th and
$i+1$st ones.  On these two factors, it acts by
\begin{align*}
T \ket{v_i}\ket{v_j} = \begin{cases}
\ket{v_j}\ket{v_i} & \text{if $i < j$},\\
(q^{-1}-q)\ket{v_i}\ket{v_j} + \ket{v_j}\ket{v_i} & \text{if $i > j$},\\
q^{-1}\ket{v_i} \ket{v_j} & \text{if $i = j$}. \end{cases}
\end{align*}
Note that when $q = 1$, we recover the the action defined by \ref{perm}.

In the rest of this thesis we consider the case where $V$ is the
representation $V^\lambda$ of $U_q(d)$  indexed by the single-box partition
$\lambda = (1)$. Using the coproduct struction on $U_q(n)$, we interpret
$V^{\otimes n}$ as a representation of $U_q(n)$ as well as $H_q(n)$. In
order for $V^{\otimes n}$ to be a representation of the algebra $U_q(d)
\otimes H_q(n)$, their respective actions must commute. This is proved
by Jimbo in \cite{J}, as well as the following result which is known as
\emph{quantum Schur-Weyl duality}.

\begin{theorem}\label{schurweyl} The space $V^{\otimes n}$ as a
representation of $U_q(d) \otimes H_q(n)$ decomposes into irreps in the
following formula
\begin{align}V^{\otimes n} \cong \bigoplus_{\lambda \vdash n, \ell(\lambda) \leq d}
V^\lambda \otimes R^\lambda\end{align}
\end{theorem}

Given that the basis of $V$ is indexed by SSYT with a single box, we think
of basis elements of $V$ as letters, and thus the natural basis vectors
of $V^{\otimes n}$ as words of length $n$. This basis can be called the
the \emph{computational} or \emph{word} basis for $V^{\otimes n}$.

With respect to the word basis for $V^{\otimes n}$ and the GTJ and YYH bases
for $V^\lambda$ and $R^\lambda$, respectively, an algorithm carrying out
the isomorphism given in \ref{schurweyl} is known as a \emph{Schur-Weyl
transform}. In the next chapter we describe a transform which decomposes
$V^{\otimes n}$ and prove that it is in fact a Schur-Weyl transform. The
transform we define has a nice recursive structure, and we prove it has
efficient time complexity in Chapter \ref{ch:schurxform}.

%% file: pieri.tex
A \emph{Pieri rule} is a formula for decomposing representations of $U_q(d)
\otimes U_q(d)$ that take the form $V^\lambda \otimes V^{(m)}$. In this
chapter we will only need a Pieri rule for the case when $m = 1$ so that
we are decomposing $V^\lambda \otimes V$.

Although the inclusion $U_q(d) \subseteq U_q(d) \otimes U_q(d)$ is not
in general a Gelfand pair, in our case of interest the branching rule is
multiplicity-free and is given in the following theorem.

\begin{theorem} Given the representation $V^\lambda$ and the representation
$V = V^{(1)}$ of $U_q(d)$, their tensor product decomposes into irreps
according to the following formula.
\begin{align}\label{pieri} V^\lambda \otimes V \cong \bigoplus_{\substack{\text{$\mu$ covers $\lambda$}\\ \ell(\lambda) \le d}} V^\mu \end{align}
\end{theorem}
An algorithm carrying out the isomorphism in equation \ref{pieri} with
respect to the GTJ bases for all $U_q(d)$ representations is called a
\emph{Pieri transform}. Interpreting the Pieri transform at the level of
tableaux, q-insertion becomes relevant. A basis element of $V^\lambda
\otimes V$ is indexed by an SSYT $t$ and a letter $i$. The correct way
of thinking of the Pieri transform is that it q-inserts $i$ into $t$ in
superposition, resulting in elements which index vectors in $V^\mu$ for
$\mu$ covering $\lambda$. A matrix entry of the Pieri transform is written
$\langle s \mid t,i \rangle$ and is non-zero only in the case when $s$
is a result of q-inserting $i$ into $t$. (This follows directly from the
behavior of the generators of $U_q(d)$.) The Pieri coefficients are also
called \emph{Wigner coefficients}.

In the rest of this section we describe formulas for $\langle s \mid
t, i \rangle$. The formulas can be found in a variety of sources, in
particular in \cite{Da}. Similar to the GTJ coefficients described in
Theorems \ref{th:gtj1} and \ref{th:gtj2}, they are complicated-looking but
in fact arise from simple combinatorial properties (e.g. axial distances)
of the tableaux associated to the basis elements.

As mentioned above, when computing $\langle s \mid t, i \rangle$, we
visualize $i$ as being q-inserted in $t$, activating a sequence of letter
bumps. The first letter considered is $i$, and it must bump a larger
letter until the last letter bumped is some letter $i_k \le d$. WLOG we
may assume that $i_k = d$, because if not we may work over the smaller
algebra $U_q(i_k)$. Thus we can visualize a sequence of letters
\begin{align}\label{bumpchain} i = i_1 < i_2 < \dots < i_k = d\end{align}
influenced by the bumping process. The Wigner coefficients factor into
a product of \emph{reduced Wigner coefficients}, one for each letter in
\ref{bumpchain}. The letter $i$ gets a special type of reduced Wigner
coefficient, which we'll refer to as type zero.

\begin{theorem}\label{th:typeonewigner} Suppose $i$ is q-inserted into the
SSYT $t$ into the box $b$ in row $r$. The corresponding type zero reduced
Wigner coefficient is given by
\begin{align} W_0(i;\lambda^{(i)}) = q^{(\text{res}(b) - x_i(t)+1)/2} \sqrt{\frac{1}{[a_{ir} + 1]} \prod_{j \le i-1} \frac{[a_{jr} - \lambda_j^{(i)}]}{[a_{jr}+1]} }  \end{align}
where $a_{jk}$ are the axial distances with respect to the horizontal
strip $\lambda^{(i)}$.
\end{theorem}

Any other letter $k$ in the chain \ref{bumpchain} is assigned a reduced
Wigner coefficient of type one. Here we have two parameters: the box $k$
inhabits in $t$ in row $r_1$ and the box $k$ gets bumped into in tableau
$s$, in row $r_2$. Then the type one Wigner coefficients are given in the
following theorem.

\begin{theorem}\label{th:typetwowigner} Suppose $i$ is q-inserted into the
SSYT $t$ and $k$ gets bumped from row $r_1$ to $r_2$ as a result of the
insertion. The corresponding type one reduced Wigner coefficient is given by
\begin{align}
W_1(k;\lambda^{(k)}) = sgn(r_1-r_2) q^{(a_{r_1r_2}-\lambda_{r_1}^{(i)})/2}
\sqrt{ \prod_{\substack{j \le k\\ j \ne r_2}} \frac{[a_{jr_1}+\lambda_{r_1}^{(k)} + 1]}{[a_{jr_2}+1]} \prod_{\substack{j \le k-1\\ j \ne r_1}}
\frac{[a_{jr_2}-\lambda_j^{(k)}]}{[a_{jr_1} + \lambda_{r_1}^{(k)} - \lambda_{j}^{(k)}]}}
\end{align}
where $a_{jj'}$ are the axial distances with respect to the SSYT $t^{(k)}$
and $sgn(0) = 1$.
\end{theorem}
Note that the exponent of $q$ given by $a_{r_1r_2} - \lambda_{r_1}^{(i)}$
in theorem \ref{th:typetwowigner} is simply the distance from the old box
$k$ inhabited in $t$ to the new box $k$ inhabits in $s$.

\begin{theorem}\label{th:wigner} The value of the Wigner
coefficient $\langle s \mid t, i \rangle$ with associated sequence
(\ref{bumpchain}) and with horizontal strips $\lambda^{(i)} =
\text{sh}(t^{(i)})\setminus\text{sh}(t^{(i-1)})$ is given by
\begin{align}\label{rwcformula} \langle s \mid t, i \rangle = W_0(i) \prod_{j = 2}^k W_1(i_j)\end{align}
\end{theorem}

\begin{example}
Let $t = \young(112,23)$ and $s = \young(112,22,3)$ so that $i = 2$ is
the letter q-inserted into $t$ to make $s$. There are two reduced Wigner
coefficients, one associated to 2, and one associated to the 3 that 2 bumps.

In the first step, 2 is added onto the second row of the tableau
$\young(112,2)$. The residue of this box is zero, $x_2(t) = 2$, and the
relevant axial distance is $a_{12} = 3$. Thus the type zero reduced Wigner
coefficient associated to the 2 is given by
$$W_0(2) = q^{-1/2} \sqrt{[2]} $$

In the second step, the 3 is bumped from its original position in the second
row to its new position in the third row. The axial distance from the old
box to the new box is given by 1. The relevant axial distances are $a_{12}
= 2$, $a_{23} = 2$ and $a_{13} = 4$.
$$W_1(3) = -q^{1/2} \frac{[4]}{[3]}\sqrt{\frac{[2]}{[5]}}$$

Therefore, the Wigner coefficient $\langle s \mid t, i \rangle$ is given
by the product
$$W_0(2) W_1(3) = - \frac{[4][2]}{[3]\sqrt{[5]}}$$
\end{example}
Theorem \ref{th:wigner} proves that the matrix entries of the Pieri transform
decompose into a product of reduced Wigner coefficients. Another way of
interpreting formula \ref{rwcformula} is the following recursive version.
\begin{align}\label{recursive}
 \langle s \mid t, i \rangle = \left\{\begin{array}{ll} W_1(d) \langle s^{(d-1)} \mid t^{(d-1)}, i \rangle & i \ne d\\ W_0(d) & i = d\end{array}\right.
 \end{align}
We then define the \emph{reduced Wigner transform} to be an algorithm
computing the $d \times d$ matrix of reduced Wigner coefficients where $t$
is fixed but $i$ and the nonzero row in $\text{sh}(s) \setminus \text{sh}(t)$
both vary.

%% file: cascade.tex
Recall we use the notation $V$ for the $d$-dimensional irrep of $U_q(d)$
indexed by the single-box partition $\lambda = (1)$. Thus we think of
basis elements of $V$ as just letters and basis elements of $V^{\otimes n}$
as words of length $n$.

For $n = 2$, we have $V^{\otimes 2}$ which can be decomposed with a single
Pieri transform seen in Section \ref{s:pieri}. For larger $n$, the Pieri
transforms can be composed (or, cascaded) to create a transform with input
space $V^{\otimes n}$. For example, when $n = 3$, we realize a decomposition
of $V^{\otimes n}$ via two sequential applications of the Pieri transform.
\begin{align*}
V^{\otimes 3} & = (V \otimes V) \otimes V\\
& \cong \left(V^{(1,1)} \oplus V^{(2)} \right) \otimes V\\
& \cong \left(V^{(1,1)} \otimes V \right) \oplus \left(V^{(2)} \otimes V\right)\\
& \cong \left(V^{(2,1)} \oplus V^{(1,1,1)} \right) \oplus \left(V^{(3)} \oplus V^{(2,1)}\right)\\
& \cong V^{(1,1,1)} \oplus \left(R^{(2,1)} \otimes V^{(2,1)}\right) \oplus V^{(3)} \end{align*}
The above description can be read as two sequential q-insertions mapping
a word of length three to a tableau with three boxes. Note that there are
two copies of $V^{(2,1)}$ in the decomposition, determined by whether the
box in the second row was added in the first or second instance of the
Pieri transform, producing a multiplicity space $R^{(2,1)}$ for the irrep
$V^{(2,1)}$. As indicated by the notation, we know the multiplicity space
$R^{(2,1)}$ is isomorphic to the irrep $R^{(2,1)}$ of the Hecke algebra
by the Schur-Weyl duality theorem \ref{schurweyl}.

By induction, a decomposition of $V^{\otimes n}$ can be achieved by
cascading $(n-1)$ Pieri transforms, which we refer to also as a Pieri
transform. This transform results in a sum of $U_q(d)$ irreps $V^\lambda$,
whose multiplicity spaces are isomorphic to $H_q(n)$ irreps $R^\lambda$
via the Schur-Weyl duality theorem. However, a priori it is unclear whether
the change-of-basis achieved by our Pieri transform is in fact identical
to the change-of-basis required by a Schur-Weyl transform. In the rest of
this section we prove that cascaded Pieri transforms compute the Schur-Weyl
transform up to sign.

We refer to the GTT basis achieved by $(n-1)$ Pieri transforms as the
Pieri basis and the basis for Schur-Weyl duality as the Schur basis,
as in Section \ref{s:schurweyl}.

\begin{theorem}\label{th:agree} The Schur transform for decomposing
$V^{\otimes n}$ by $(n-1)$ cascaded Pieri transforms is a Schur-Weyl
transform as defined in Section \ref{s:schurweyl} up to signs.
\end{theorem}

\begin{proof} We first argue the $q=1$ case, so that we can work with groups.
Beginning with the group $S(n) \times U(d)$, in both the Pieri and Schur-Weyl
decompositions we have subgroup flags that reach first $U(d)$ to produce
a summand of irreps $V^\lambda$. From this point we continue with the
subgroup flag defined by the branching rule in theorem \ref{th:branch},
to yield the standard GTT basis of $V^\lambda$.

In the Schur basis the subgroup flag is defined by
$$U(d) \subseteq S(2) \times U(d) \subseteq S(3) \times U(d) \subseteq
    \cdots \subseteq S(n) \times U(d)$$
and in the Pieri basis the subgroup flag is defined by
$$U(d) \subseteq U(d)^2 \subseteq U(d)^3 \subseteq \cdots \subseteq U(d)^n$$
In the Pieri flag, we mean more precisely that within the group $U(d)^k$,
the first factor of $U(d)$ should act on the first $n-k+1$ tensor factors
of $V$ diagonally, while the $j$th factor of $U(d)$ for $j \ge 2$ should
act on the $(n-k+j)$th tensor factor of $V$.  In other words, we can define
the desired embedding $U(d)^k \subseteq U(d)^{k+1}$ by the map
$$(\Delta \times \id^{k-1}):U(d)^k \to U(d)^{k+1},$$
where
$$\Delta:U(d) \to U(d)^2$$
is the standard diagonal embedding, and $\id$ is the identity.

We claim that the direct sum decompositions induced by the two subgroup flags
become equal when they reach $U(d)$ and irreps $V^\lambda$ of this group.
Since the subgroup flags coincide below $U(d)$, the decomposition must remain
equal afterwards.  To prove the claim, we consider the partially ordered set
$P$ of groups $S(j) \times U(d)^k$, with $j+k \le n+1$,
as shown in \fig{f:poset}.

$P$ has a unique minimal element, $U(d)$, and
$n$ maximal elements.  It also has $2^n$ maximal chains that connect $U(d)$
to some maximal element, all of length $n+1$; each step of such a chain,
in the figure, can either be down or to the right.  We claim that each of
these maximal chains are all locally multiplicity free on $V^{\tensor n}$,
and that the partial decompositions all coincide when they reach $U(d)$.

To see that each chain $c \subset P$ produces a GTT line basis, we first
decompose $V^{\tensor n}$ as a representation of a maximal group $S(k)
\times U(d)^{n-k+1}$.  By \ref{schurweyl} we obtain the following isomorphism:
\begin{align*}
V^{\tensor n} & \cong V^{\otimes k} \otimes V^{\otimes n-k}\\
& \cong \left(\bigoplus_{\lambda \vdash k} R_\lambda \tensor V_\lambda\right)
    \tensor V^{\tensor n-k}
\end{align*}
This is multiplicity free.  Each subsequent step of $c$ is one of the
inclusions
\begin{align*}
S(j-1) \times U(d)^k & \subseteq S(j) \times U(d)^k\\
S(j) \times U(d)^{k-1} & \subseteq S(j) \times U(d)^k \end{align*}
Both of these inclusions are Gelfand pairs by Theorems \ref{th:branch}
and \ref{th:heckebranch}, and the structure of irreps of the direct product
of two groups.

To see that the decompositions coincide, we consider two types of
moves on chains in $P$:  a triangle move that changes the last step between
horizontal and vertical, and a square move that switches a horizontal
and vertical step lower in the chain.  The triangle move relates
two chains $c_1, c_2 \subseteq P$ that agree except at the three groups:
$$\begin{matrix}
S(k) \times U(d)^{n-k} & \subseteq & S(k+1) \times U(d)^{n-k} \\
\vsubseteq \\ S(k) \times U(d)^{n-k+1} \end{matrix}$$
We claim that the decomposition of $V^{\otimes n-k}$ is multiplicity free using
the chain $c_3 = c_1 \cap c_2$, which begins directly with
$S(k) \times U(d)^{n-k}$.  By \ref{schurweyl} and either \thm{th:branch}
or \thm{th:heckebranch},
we obtain
$$V^{\tensor n} \cong \left(\bigoplus_{\substack{\lambda \vdash k, \mu \vdash k+1
    \\ \mu \text{covers} \lambda}} R^\lambda \tensor V^\mu\right) \tensor V^{\tensor n-k-1},$$
which is multiplicity free.  Since $c_1$ and $c_2$ each yield the same
decomposition as $c_3$, they yield the same decomposition as each other.

Likewise suppose that $c_1, c_2 \subseteq P$ differ by a square move:
$$\begin{matrix} S(j) \times U(d)^k & \subseteq & S(j+1) \times U(d)^k \\
\vsubseteq && \vsubseteq \\
S(j+1) \times U(d)^k & \subseteq & S(j+1) \times U(d)^{k+1} \end{matrix}.$$
We claim that $c_3 = c_1 \cap c_2$ is again locally multiplicity free,
which implies that $c_1$ and $c_2$ must each yield the same decomposition
as $c_3$.  At the lower left corner, a single summand which is an irrep
of $S(j+1) \times U(d)^{k+1}$ will in general have the form  $(R_\lambda
\tensor V_\mu) \tensor V^{\tensor k}$.  Then its restriction to $S(j) \times
U(d)^k$ is multiplicity free, by applying \thm{th:heckebranch} to
$R_\lambda$ and \thm{th:branch} to $V_\mu \tensor V$.

It is easy to see that all maximal chains in $P$ are connected by
square and triangle moves.  This yields the result when $q=1$.

When $q$ is positive (or more generally, when $q$ is not a root of unity)
we can follow the same argument, except that we replace $S(j) \times
U(d)^k$ by $H_q(j) \tensor U_q(d)^{\tensor k}$.  The replacement yields
well-defined algebra actions by \ref{schurweyl}, and the argument still
works because it relies on the same multiplicity free structures.

The line basis agreement extends to unique vector bases up to sign.
Because the representations considered are all bar representations, the same
Gelfand-Tsetlin constructions yield unique real line bases.  Then, because
the representations are all *-representations, the real line bases can be
refined to bases of real unit vectors.  These vectors are then unique up
to sign.

\end{proof}

\begin{figure*}
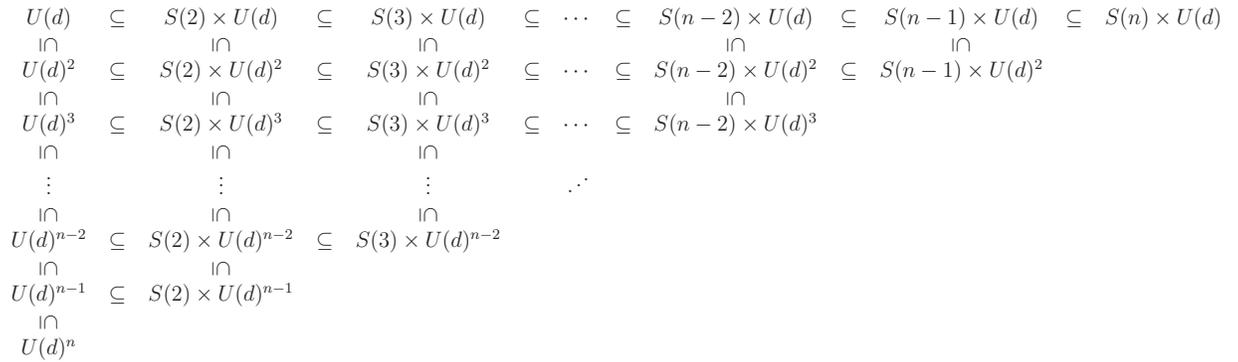

\begin{center}
\scalebox{.7}{
\begin{tabular*}{10in}{ccccccccccccc}
$U(d)$ & $\subseteq$ & $S(2) \times U(d)$ & $\subseteq$ & $S(3) \times U(d)$ & $\subseteq$ &
    $\cdots$ & $\subseteq$ & $S(n-2) \times U(d)$ & $\subseteq$ & $S(n-1) \times U(d)$
    & $\subseteq$ & $S(n) \times U(d)$ \\
$\vsubseteq$ && $\vsubseteq$ && $\vsubseteq$ &&&& $\vsubseteq$ && $\vsubseteq$ \\
$U(d)^2$ & $\subseteq$ & $S(2) \times U(d)^2$ & $\subseteq$ & $S(3) \times U(d)^2$ &
    $\subseteq$ & $\cdots$ & $\subseteq$ & $S(n-2) \times U(d)^2$ & $\subseteq$ &
  $S(n-1) \times U(d)^2$ \\
$\vsubseteq$ && $\vsubseteq$ && $\vsubseteq$ &&&& $\vsubseteq$ \\
$U(d)^3$ & $\subseteq$ & $S(2) \times U(d)^3$ & $\subseteq$ & $S(3) \times U(d)^3$ &
    $\subseteq$ & $\cdots$ & $\subseteq$ & $S(n-2) \times U(d)^3$ \\
$\vsubseteq$ && $\vsubseteq$ && $\vsubseteq$ \\
$\vdots$ && $\vdots$ && $\vdots$ &&  $\adots$ \\
$\vsubseteq$ && $\vsubseteq$ && $\vsubseteq$ \\
$U(d)^{n-2}$ & $\subseteq$ & $S(2) \times U(d)^{n-2}$ & $\subseteq$
    & $S(3) \times U(d)^{n-2}$ \\
$\vsubseteq$ && $\vsubseteq$ \\
$U(d)^{n-1}$ & $\subseteq$ & $S(2) \times U(d)^{n-1}$ \\
$\vsubseteq$ \\
$U(d)^n$
\end{tabular*}}
\end{center}
\caption{The poset $P$ of group inclusions}
\label{f:poset}
\end{figure*}

%% file: crystal.tex
In quantum algebra the limit $q = 0$ is referred to as the \emph{crystal
limit}. The quantum groups and Hecke algebras we have considered do not
have well-defined algebra structures at $q = 0$. However, it can be useful
to look at the transforms in the crystal limit, as they are still linear
maps between vector spaces, if not representation isomorphisms. In this
section we describe the behavior of the Wigner coefficients defined in
the Pieri transform in Section \ref{s:pieri} in the crystal limit. These
theorems can be found in paper \cite{Da}.

\begin{theorem}\label{th:crystal1} The type zero Wigner coefficient
$W_0(r;\lambda^{(i)})$ is zero in the limit $q = 0$ except when for every $j$
between $r$ and $i-1$ we have $\lambda^{(i)}_{j+1} = \lambda^{(i-1)}_j$. The
type zero Wigner coefficient $W_0(r;\lambda^{(i)})$ is zero in the limit $q
= \infty$ except when $r = 1$, and in this case it's one.  \end{theorem}
\begin{theorem}\label{th:crystal2} The type one Wigner coefficient
$W_1(r_1;r_2;\lambda^{(i)})$ is zero in the limit $q = 0$ for every choice
except when $r_1 = r_2$ or when $r_1 > r_2$ and $\lambda^{(i)}_{j+1}
= \lambda^{(i-1)}_j$ for j between $r_2$ and $r_1-1$, and in both these
cases it's 1. The type one Wigner coefficient $W_1(r_1;r_2;\lambda^{(i)})$
is zero in the limit $q = \infty$ for every choice except when $r_2 = r_1 +
1$, and in this case it's $-1$.  \end{theorem} At the level of tableaux,
theorems \ref{th:crystal1} and \ref{th:crystal2} can be interepreted in
the following way.

In the limit $q = 0$, the only time a letter can be inserted into a tableau
without bumping another letter is when it is added to the very first column
that does not contain an $i$. Also, the only time a letter can be bumped,
it gets inserted into the very next column. This is a description of the
dual RSK insertion algorithm.

In the limit $q = \infty$ the only time a letter can be inserted into
a tableau without bumping another letter is when it can be added to the
first row, and the only time a letter can be bumped, it gets inserted into
the following row. This is a description of the RSK insertion algorithm
together with a bumping sign described in section \ref{s:qinsertion}.

Therefore, in the crystal limits $q = 0$ and $q = \infty$ the transform
defined in \ref{pieri} is the dual RSK and RSK insertion algorithm with
bumping sign, respectively. In this sense, RSK and dual RSK are `classical'
versions of quantum insertion (not to be confused with the other notion of
`classical' in this setting, i.e. $q = 1$).

Extending Theorems \ref{th:crystal1} and \ref{th:crystal2} to cascaded
Pieri transforms, we conclude that in the crystal limits, the Schur-Weyl
transform computes the well-known RSK bijections on finite sets described
in Section \ref{s:qinsertion}.

%% file: qprob.tex
In this section we describe some elements of quantum probability and their
relation to quantum algorithms. A thorough treatment of this material can
be found in the book \cite{NC}.

We fix the \emph{computational basis} of $\C^d$ to consist of the orthonormal
vectors $\ket{1},\dots,\ket{d}$, and write all matrices in this basis. The
algebra $M_d$ is the set of all $d \times d$ matrices in this basis. (The
more standard numbering in computer science is from $0$ to $d-1$; we use
the mathematicians' numbering which is more standard in combinatorial
representation theory.)

When $d = 2$, $M_2$ is called a \emph{qubit}, and represents the probability
space corresponding to a quantum particle with two basis states, such
as an electron which can be measured as either spin-up or spin-down. For
larger $d$, $M_d$ is called a \emph{qudit}. Joint systems are constructed
by tensoring qudits. So, for example, a pair of qubits is $M_2 \otimes
M_2 \cong M_4$.

The \emph{state} of a qudit $M_d$ is defined to be a positive and
normalized dual vector on $M_d$. There is an isomorphism between $M_d$ and
its dual space so we can view dual vectors of $M_d$ as elements of $M_d$
itself. Under this isomorphism, $\rho \in M_d$ acts as a dual vector on
$M_d$ according to the formula
\begin{align}\label{stateformula} \rho(a) = \text{Tr}(\rho a)\end{align}

The \emph{pure} states of $M_d$ are of the form $\rho_{\ket{\psi}} =
\ket{\psi}\bra{\psi}$, where $\ket{\psi}$ is a normalized vector in
$\C^d$. Following Equation \ref{stateformula}, its action on an element
$b \in M_d$ is given by the inner product
$$\rho_{\ket{\psi}} (b) = \text{Tr}(\ket{\psi}\bra{\psi} b) = \bra{\psi} b \ket{\psi}.$$ Thus pure states are indexed by normalized vectors in $\C^d$.

The self-adjoint elements in $M_d$, i.e. those that satisfy $a^* = a$,
are called \emph{observables} or \emph{measurables}. The \emph{idempotent}
observables, sometimes called \emph{events}, satisfy $a^2 = a$, and are
interpreted as measuring whether the qudit is in the state $a$. By the
Spectral theorem, any observable $a$ can be decomposed into a sum of
idempotent observables:
\begin{align}\label{spec}
a = \sum_{\lambda \in \sigma(a)} \lambda a_{\lambda},
\end{align}
where $\sigma(a)$ is the spectrum of $a$, and $a_{\lambda}$ is the projection
operator onto the eigenspace defined by $\lambda$. The Spectral theorem
guarantees a choice of projection operators $a_\lambda$ which are pairwise
orthogonal. The probability that the observable $a$ measures $\lambda$ is
\begin{align}\label{measprob}\text{Prob}[a = \lambda] = \text{Tr}(\rho a_\lambda)\end{align}
and the state of the qudit passes to the conditional state $\rho_\lambda =
\rho_{\ket{\phi}}$, which is defined by
\begin{align}\label{conditionalpure}\ket{\phi} = \frac{a_\lambda \ket{\psi}}{\sqrt{ \bra{\psi} a_\lambda \ket{\psi}}}\end{align}

We will mainly use the special case $a = \sum_k \ket{k} \bra{k},$ which
is a \emph{complete} measurement in the computational basis. Expanding
the vector $\ket{\psi} = \sum_k \psi_k \ket{k}$, we have
$$\rho_{\ket{\psi}} = \sum_{j,k} \rho_{j,k} \ket{j}\bra{k},$$
where $\rho_{j,k} = \psi_j^* \psi_k$. Therefore, following Equation
\ref{measprob}, the probability that the state of the qudit is measured
to be $\ket{k}$ is $\rho_{k,k} = | \psi_k |^2$, and following Equation
\ref{conditionalpure}, the state then passes to $\rho_k = \ket{k}
\bra{k}$. So, after a complete measurement the state of a qudit is defined
by one of the basis vectors $\ket{k}$, which can be used as an `answer'
to a computational question.

The state of a qudit can also undergo reversible unitary \emph{evolution}
$E$ given by conjugations $E(\rho) = U\rho U^{-1}$ where $U \in PSU(d)$,
the space of projective unitary maps. If $\rho_{\ket{\psi}}$ is a pure
state, then
$$E(\rho_{\ket{\psi}}) = U \rho_{\ket{\psi}} U^{-1} = \rho_{U\ket{\psi}},$$
is also pure and defined by the action of an element in $PSU(d)$.

Pure states, measurement operators, and unitary evolution are the basic
elements necessary to define quantum algorithms, which we describe in
Secction \ref{s:qalgs}.

%% file: qalgorithms.tex
In this section we review basic quantum algorithms. For a more complete
introduction, refer again to \cite{NC}.

In the context of quantum computing, unitary operators acting on qudit spaces
are called \emph{quantum gates}. In the usual interpretation, a quantum gate
acting on $m$ qudits can act on $n \ge m$ qudits by acting by $U$ on $m$
of the qudits and the identity on the other $n-m$ qudits. A \emph{quantum
circuit} is a composition of quantum gates. The \emph{time complexity} of
a quantum circuit is the number of quantum gates in its decomposition. We
require that a family of quantum circuits be \emph{uniform}, meaning there is
a classical algorithm to compute the decomposition of the quantum circuits
into quantum gates. The time complexity of this decomposition algorithm
is counted toward the total time complexity of the quantum algorithm.

Because the set of quantum gates is infinite, but a quantum computer would
have access to a finite number of quantum gates, it's generally not possible
to construct an exact quantum circuit for calculating a given unitary
transformation. In other words, some approximation will usually be necessary.

A finite set of gates $G$ that acts on at most $m$ qudits of size $d$
is called \emph{universal} if it generates a dense subgroup of $U(d^m)$
for some $m \ge 2$; it consequently densely generates $U(d^n)$ for any $n
\ge m$.  In other words, a set $G$ of quantum gates is universal if every
unitary operator $A \in U(d^n)$ can be approximated by a quantum circuit
composed of elements from $G$.

The Solovay-Kitaev theorem \cite{Dawson} establishes that any operator $A
\in U(d)$ can be approximated by words in a universal gate set with time
and gate complexity $\poly(d,\log \eps^{-1})$, where $\epsilon$ is the
error of the approximation.  This is provided that the matrix entries of
the gates and the matrix entries of $A$ can be approximated with the same
time complexity.  Thus, up to polylogarithmic overhead, any universal gate
set is equivalent to all unitary operators in $U(d)$ whose matrix entries
can be computed quickly.  Often this theorem is used for fixed values of
$d$, but the algorithm is constructive and it is easy to establish that
the gate complexity (and classical time complexity to choose the gates)
is polynomial in $d$ as well. The Solovay-Kitaev theorem establishes that
this time complexity is independent of the universal gate set up to a
polylogarithmic factor (assuming that the matrices of the gates can be
computed quickly).

Both classically and quantumly, an \emph{efficient} algorithm is one
which has polynomial time complexity. The class of decision problems with
efficient classical probabilistic algorithms is called \emph{BPP}, for
bounded-error polynomial time. The class of decision problems with efficient
quantum algorithms is analogously called \emph{BQP}, for bounded-error
quantum polynomial time. The class BQP contains the class BPP, meaning
that quantum algorithms can efficiently solve any decision problems that
classical algorithms can solve. Whether or not BQP is strictly larger
than BPP remains an open problem. However, there are decision problems
provably in BQP that are not provably in BPP, such as the factoring problem
investigated in \cite{shor}.

%% file: schurxform.tex
In this section we describe the time complexity of computing first a
single Pieri transform, and and then the cascade of Pieri transforms which
calculates a Schur-Weyl transform.

In Equation \ref{rwcformula}, we calculated the matrix entry
$\langle s \mid t, i \rangle$ of the Pieri transform in terms of
reduced Wigner coefficients. Using Formulas \ref{th:typeonewigner} and
\ref{th:typetwowigner} with \ref{rwcformula} allow us to calculate all
the matrix entries of the Pieri transform, but at this point a direct
application of the Solovay-Kitaev theorem is not quite sufficient, since
the entire Pieri operator has a matrix of size $d|\SSYT(\lambda)|$, which
is larger than $\poly(n,d)$.

This is why we use the recursive version of \ref{rwcformula} given by
\ref{recursive} to define the reduced Wigner transform as the $d \times d$
matrix of reduced Wigner coefficients where $t$ is fixed. As described in
\cite{BCH}, given an input $t$, the coefficients in the reduced Wigner
transform are calculated by a classical algorithm and then the reduced
Wigner transform is calculated by a quantum algorithm.

\begin{theorem}\label{th:rwt} There is a quantum algorithm for computing
the controlled reduced Wigner transform with time complexity $\poly(n,d,\log
\eps^{-1})$, where $\eps$ is the desired accuracy. \end{theorem}

Combining the recursive formula \ref{recursive} with Theorem \ref{th:rwt},
we obtain the following result.

\begin{theorem}\label{th:qpieri}There is a quantum algorithm to compute the
quantum Pieri transform with accuracy $\eps$ with time complexity $poly(n,d,
\log \eps^{-1})$, where $\eps$ is the desired acuracy. \end{theorem}

As proved in section \ref{s:cascade}, a Schur-Weyl transform can be
computed up to sign using $(n-1)$ cascaded Pieri transforms. We use this
fact combined with theorem \ref{th:qpieri} to calculate the total time
complexity of the Schur-Weyl transform given in the following theorem.

\begin{theorem}\label{th:qschur} There is a quantum algorithm for computing
the cascaded Pieri transform which is a Schur-Weyl transform on $V^{\tensor
n}$ with time complexity $\poly(n,d,\log(\epsilon^{-1}))$, where $\epsilon$
is the desired accuracy. \end{theorem}

Since Theorems \ref{th:qpieri} and \ref{th:qschur} are entirely based on
unitary operators for any real number $q > 0$, and since unitary groups
are compact, we can expect the Schur algorithm to have well-defined limits
at $q = 0$ and $q = \infty$, provided that we keep $q$ real and positive.
This expectation turns out to be correct.  As proved in \cite{Da} and
discussed in Section \ref{s:crystal}, in the limit $q = 0$ the transforms
converge to permutations matrices, while in the limit $q = \infty$, they
converge to signed permutation matrices.

Therefore, at $q = 0$ and $q = \infty$ our algorithm is a unitary version
of the dual RSK and RSK algorithm combined with a bumping sign algorithm.

%% file: conclusion.tex
We conclude this thesis by discussing some directions for further research
into the topic of Schur transforms.

First of all, it is still open whether there is an efficient algorithm for
the Schur-Weyl transform for any value of $q$, including $q = 1$, which
is jointly polynomial in $n$ and $\log d$, in other words polynomial in
the input qubit length $n(\log d)$.

A second line of investigation involves finding algorithms for other
transforms of quantum algebras. In particular, the algorithm in \cite{Beals}
which is an efficient computation for the quantum Fourier transform for
the symmetric group $S(n)$ could possibly be generalized to a quantum
Fourier-like transform decomposing the regular representation of the Hecke
algebra $H_q(n)$.

The quantum Fourier transform decomposes a representation of $S(n) \times
S(n)$, while the Schur transform decomposes a representation of $S(n)
\times U(d)$. There is a known algebra isomorphism for decomposing an
analogous representation of $U(d) \times U(d)$ known as Howe duality,
which is extended to the quantum algebra $U_q(d)$ in \cite{zhang}. Whether
there is an efficient quantum transform for computing a version of Howe
duality for any value of $q$, including $q = 1$, is open. A possible goal
is to generalize all three algorithms (Fourier, Schur, and Howe) into a
single algorithm.

Finally, one may investigate the types of applications $q$-deformed
quantum algorithms can solve. One of the main reasons that quantum Fourier
transforms have been studied so extensively is their link to hidden subgroup
problems which are in turn linked to interesting computational problems
such as factoring and graph isomorphism. In \cite{BCH} some applications
of the Schur transform for $q = 1$ are proposed. The types of computational
problems quantum algorithms for decomposing quantum algebra representations
help to solve remain to be investigated.